\documentclass[11pt]{article}

\usepackage{fullpage}
\usepackage{amsmath,amssymb,amsfonts,xspace,graphicx,relsize,bm,mathtools,xcolor}
\usepackage{mathrsfs}
\usepackage{lipsum}
\usepackage[pagebackref]{hyperref}
\usepackage[normalem]{ulem}

\usepackage[thmmarks]{ntheorem}

\usepackage{multirow}
\usepackage{cleveref}
\usepackage{thm-restate}
\usepackage{array}
\usepackage{soul}

\usepackage[usestackEOL]{stackengine}

\usepackage{url}
\def\01{\{0,1\}}

\newcommand{\eps}{\varepsilon}
\newcommand{\ket}[1]{|#1\rangle}

\newcommand{\norm}[1]{\mbox{$\parallel{#1}\parallel$}}

\newcommand{\Cc}{{\mathcal C}} %
\newcommand{\Exp}{\mathbb{E}}

\newcommand{\A}{\ensuremath{\mathcal{A}}}
\newcommand{\F}{\ensuremath{\mathcal{F}}}
\newcommand{\RF}{\ensuremath{\mathcal{RF}}}
\newcommand{\De}{\ensuremath{\mathcal{D}}}
\newcommand{\R}{\ensuremath{\mathbb{R}}}

\newcommand{\Mod}[1]{\ (\mathrm{mod}\ #1)}
\newcommand{\mbf}[1]{\ensuremath{\mathbf{#1}}}

\newcommand{\biaslearning}{\beta}

\newcommand{\query}{\mathsf{q}}
\DeclareMathOperator{\Kpub}{K_{pub}}
\DeclareMathOperator{\Kpriv}{K_{priv}}
\DeclareMathOperator{\ideal}{ideal}
\DeclareMathOperator{\real}{real}
\newcommand{\Hi}{\ensuremath{\mathcal{H}}}

\DeclareMathOperator{\poly}{poly}

\DeclareMathOperator{\polylog}{polylog}
\DeclareMathOperator{\superpoly}{superpoly}

\newcommand{\Fe}{\ensuremath{\mathcal{F}}}
\newtheorem{result}{Result}
\newtheorem{theorem}{Theorem}[section]
\newtheorem{definition}[theorem]{Definition}
\newtheorem{fact}[theorem]{Fact}

\newtheorem{lemma}[theorem]{Lemma}

\newtheorem{corollary}[theorem]{Corollary}

\newtheorem{claim}[theorem]{Claim}
\newtheorem{assumption}{Assumption}
  \newcommand{\pmset}[1]{\{-1,1\}^{#1}} %

\def\01{\{0,1\}}

\newcommand{\pr}{\operatorname{Pr}}
\def\ints{\mathbb{Z}}

\newcommand{\supp}{\mathrm{supp}}

\DeclareMathOperator{\negl}{negl}

\newcommand{\NCz}{\ensuremath{\mathrm{NC}^0}}

\newcommand{\ACz}{\ensuremath{\mathsf{AC}^0}}

\newcommand{\TCz}{\ensuremath{\mathsf{TC}^0}}

\newcommand{\LWE}{\ensuremath{\mathsf{LWE}}}
\newcommand{\LWR}{\ensuremath{\mathsf{LWR}}}

\newcommand{\RLWE}{\ensuremath{\mathsf{RLWE}}}

\newcommand{\RLWR}{\ensuremath{\mathsf{RLWR}}}

\newcommand{\AND}{\ensuremath{\mathsf{AND}}}
\newcommand{\OR}{\ensuremath{\mathsf{OR}}}
\newcommand{\NOT}{\ensuremath{\mathsf{NOT}}}

\newcommand{\LWEPKE}{\ensuremath{\mathsf{LWE-PKE}}}

\newcommand{\RLWEPRF}{\ensuremath{\mathsf{RLWE\!-\! PRF}}}

\newcommand{\oBRPRF}{\ensuremath{\mathsf{1B\!-\! RPRF}}}
\newcommand{\RLWRPRF}{\oBRPRF}

\newcommand{\indic}[1]{\left\llbracket #1 \right\rrbracket}
\newcommand{\Adv}{\mathsf{Adv}}
\DeclareMathOperator{\msb}{msb}

\def\complex{\mathbb{C}}
\def\natural{\mathbb{N}}
\def\ints{\mathbb{Z}}

\usepackage{array}
\newcolumntype{M}[1]{>{\centering\arraybackslash}m{#1}}

\usepackage{kpfonts}
\newenvironment{proofof}[1]{\par\noindent{\bf Proof of #1.}\quad}{\hfill  $\Box$}
\newenvironment{sproof}{\par\noindent\textbf{Proof Sketch.}\quad}{\hfill $\Box$} %

\usepackage[margin=1in]{geometry}
\hypersetup{
    colorlinks,
    linkcolor={blue!100!black},
    citecolor={blue!100!black},
}
\usepackage{comment}

\newenvironment{proof}
{\noindent {\bf Proof. }}
{{\hfill $\Box$}\\
\smallskip}

\newcommand{\U}{\ensuremath{\mathcal{U}}}
\newcommand{\DNF}{\ensuremath{\mathsf{DNF}}}
\newcommand{\EX}{\mbox{\rm EX}}
\newcommand{\QEX}{\mbox{\rm QEX}}
\newcommand{\MQ}{\mbox{\rm MQ}}
\newcommand{\QMQ}{\mbox{\rm QMQ}}
\newcommand{\MAJ}{\mbox{\rm MAJ}}
\newcommand{\enc}{\mbox{\rm Enc}}
\newcommand{\dec}{\mbox{\rm Dec}}

\renewcommand\thmcontinues[1]{Formal Statement}

\newcommand{\torp}[2]{\texorpdfstring{#1}{#2}}

 \definecolor{greenn}{rgb}{0.2,0.6,0.2}
\definecolor{bluue}{rgb}{0.7,0,0.3}

\begin{document}

\setlength{\abovedisplayskip}{1pt}
\setlength{\belowdisplayskip}{1pt}
\interfootnotelinepenalty=10000

\title{Quantum hardness of learning shallow classical circuits}
\author{
	Srinivasan Arunachalam\thanks{Center for Theoretical Physics, MIT. Funded by the MIT-IBM Watson AI Lab under the project {\it Machine Learning in Hilbert space}. arunacha@mit.edu} 
	\and
	Alex B. Grilo~\thanks{CWI and QuSoft, Amsterdam, The Netherlands. Supported by ERC Consolidator Grant 615307-QPROGRESS. alexg@cwi.nl}
	\and
	Aarthi Sundaram~\thanks{Joint Center for Quantum Information and Computer Science, University of Maryland, USA. Supported by the Department of Defense. aarthims@umd.edu}
}
\date{}
\maketitle

\begin{abstract}
{In this paper, we study the quantum learnability of constant-depth
classical circuits under the uniform distribution and in the
distribution-independent framework of PAC learning. 
In order to attain our results, we 
establish connections between quantum learning and quantum-secure
cryptosystems. We then achieve the following
results. \medskip

	\begin{enumerate}
    \item \textbf{Hardness of learning $\ACz$ and $\TCz$ under the uniform
      distribution.} Our first result concerns the concept class $\TCz$ (resp.~$\ACz$), the class of
      constant-depth and polynomial-sized circuits with unbounded fan-in 
      majority gates (resp.~$\mathsf{\AND}, \mathsf{OR}, \mathsf{NOT}$ gates). We show the following:  
           \vspace{5pt}
      \begin{itemize}
       \item if there exists no quantum (quasi-)polynomial-time algorithm to solve the
         Ring-Learning with Errors ($\RLWE$) problem, then there exists no
         (quasi-)polynomial-time
      quantum learning algorithm for $\TCz$; and
      \vspace{5pt}
      
      	\item if there exists no $2^{O(d^{1/\eta})}$-time quantum algorithm to solve
          $\RLWE$ with dimension $d  = O(\polylog n)$ (for every constant $\eta > 2$), then there exists no $O(n^{ \log^{\nu} n} )$-time quantum learning algorithm for $\poly(n)$-sized $\ACz$ circuits (for a constant $\nu>0$), matching the \emph{classical} upper bound of Linial et al.~\cite{LMN93}
      \end{itemize}
  \vspace{0.8em}
  
     where the learning algorithms are under the uniform distribution (even with access to
      quantum membership queries). The main technique in these results uses an
      explicit family of pseudo-random functions that are believed to be quantum-secure to
      construct concept classes that are hard to learn quantumly under the
      uniform~distribution.
		  \medskip
		
	  \item \textbf{Hardness of learning $\TCz_2$ in the PAC setting.} Our second
      result shows that if there exists no quantum polynomial-time algorithm for
      the $\LWE$ problem, then there exists no polynomial-time quantum-PAC learning
      algorithm  for the class $\TCz_2$, i.e., depth-$2$ $\TCz$ circuits.
      The main technique in this result is to establish a connection between the quantum security of public-key cryptosystems and the learnability of a concept class that consists of decryption functions of the~cryptosystem.
	\end{enumerate}
	\medskip

 Our results show that quantum resources do not give an \emph{exponential}  improvement to learning constant-depth polynomial-sized neural networks. This also gives a strong (conditional) negative answer to one of the
``Ten Semi-Grand Challenges for Quantum Computing Theory" raised by Aaronson~\cite{Aaronson-blog}.	
}
\end{abstract}

\pagenumbering{gobble}
\clearpage
\pagenumbering{arabic}
\setcounter{page}{1}

\setcounter{tocdepth}{2}
\tableofcontents

\section{Introduction}
Machine learning is a diverse field of research with many real-world
applications and has received tremendous attention in the past decade.  From a
theoretical perspective, since the seminal paper of
Valiant~\cite{valiant:paclearning}, there has been a lot of theoretical effort
in considering different learning models that formalize what we mean by {\em
learning} and understanding which
problems can (or cannot) be \emph{efficiently} learned within  these models.

In recent years, due to the considerable development of quantum computation (both on the
theoretical and experimental fronts), there has been an increased focus on understanding the tasks for which quantum computers can offer
speedups. Machine learning tasks have emerged as a candidate in this respect. To this end, results in \emph{quantum learning theory} aim to  identify learning problems for which 
quantum computers \emph{provably}  provide a (significant) advantage.

More concretely, in learning theory, the goal is to devise a learning algorithm
(or learner) for a set of functions which is called a \emph{concept
class}. The functions in the concept class $\Cc$ are referred to as \emph{concepts}. In this paper, we will consider (without loss of generality) concepts that are Boolean functions $c:\01^n\rightarrow \01$.  The learner is provided with \emph{examples} of the form $(x, c(x))$, where $c$ is an \emph{unknown} concept lying in $\Cc$ and $x$
is picked uniformly at random from $\01^n$ and the goal of the learner is to
\emph{learn $c$}, i.e., it should output a hypothesis~$h$ that is close to
$c$.\footnote{More formally, this is referred to as the uniform-distribution
learning setting. In learning theory there are several variations of this model
that we skip for the sake of simplicity.  See \Cref{sec:intro-learning}
for a brief introduction and \Cref{sec:learningmodels} for more details.}
We say that a \emph{learner $\A$ learns $\Cc$}, if for every $c\in \Cc$, $\A$ learns $c$.  In learning theory, the intent is to devise \emph{efficient} learning algorithms for an interesting concept class $\Cc$, i.e., the learner should use few examples and not too much time in order to learn $\Cc$. 

In {\em quantum} learning theory, the goal is still to learn concept classes
$\Cc$, but now with the access to quantum resources. Bshouty and Jackson~\cite{bshouty:quantumpac} introduced a quantum learning model wherein the learner is a \emph{quantum} algorithm and is provided with \emph{quantum examples} of the~form
$$
\frac{1}{\sqrt{2^n}}\sum_{x \in\01^n} \ket{x,c(x)}
$$ 
for some concept $c\in \Cc$. In order to see why quantum examples generalize
classical examples, observe that a quantum learner could choose to measure the
quantum example in the computational basis, which results in a classical example
$(x,c(x))$ for a uniformly random $x\in \01^n$.  The advantage of quantum learning  usually comes from the fact that 
one can perform arbitrary unitary operations on these quantum examples, 
enabling one to improve sample or time complexity for learning the concept class $\Cc$.

The first example of a quantum advantage for a learning problem was showed by
Bernstein and Vazirani~\cite{bernstein&vazirani:qcomplexity}. They showed how to learn the concept class of linear functions $\Cc=\{c(x)=\sum_i s_i x_i \mod 2: s\in \01^n\}$ with a constant
number of quantum examples. Classically, in order to learn this concept class
efficiently, it is necessary and sufficient to obtain $\widetilde{\Theta}(n)$
examples. With these many examples a classical learner can use Gaussian elimination and learn
the unknown target concept in polynomial time.

Subsequently, Bshouty and Jackson~\cite{bshouty:quantumpac} showed that the
class of Boolean functions that can be represented as polynomial-sized $\DNF$ formulas can be learned in quantum polynomial time.\footnote{A $\DNF$ formula is a disjunction of conjunctions of variables and their negations.} A crucial ingredient for their quantum learning algorithm was the ability to perform Fourier sampling using quantum examples (which we discuss in \Cref{sec:poweroffouriersampling}). Classically, Verbeurgt~\cite{verbeurgt:learningdnf} showed that $\DNF$s can be learned in quasi-polynomial time using classical examples and this has remained the state-of-the art for the last thirty years! This emphasizes the power of quantum examples for learning. There have been many other instances (which we discuss in \Cref{sec:poweroffouriersampling}) where quantum examples give an advantage for quantum learning algorithms. 
A natural follow-up question to the quantum-efficient learnability of $\DNF$ formulas, i.e., depth~-$2$ circuits, is:
\begin{quote} 
	Can the concept class of \emph{shallow, i.e., constant-depth, classical circuits} be learned more \emph{efficiently} using \emph{quantum resources} as compared to classical resources?
\end{quote}
In particular, are there efficient (i.e., polynomial-time) quantum learning
algorithms for all of the Boolean functions that can be represented by $\ACz$ circuits, i.e., constant-depth circuits with {\em unbounded} fan-in $\AND, \OR, \NOT$ gates? More ambitiously, can we quantum-efficiently learn $\TCz$, i.e., constant-depth circuits with majority gates.\footnote{A majority gate on $n$ bits outputs $1$ if $\lfloor n/2\rfloor $+1 bits evaluate to $1$ and $0$ otherwise.} This question was raised by Aaronson~\cite{Aaronson-blog} as one of the \emph{``Ten Semi-Grand Challenges for Quantum Computing Theory"}.

We notice that understanding the
learnability of $\TCz$,
apart from being theoretically important, also sheds light on the question: can
quantum computers help learn the weights of neural networks faster? It is
well-known~\cite{maas:sigmoidboolean,goldmann:thresholdneural} that
constant-depth polynomial-sized feed-forward neural networks, where the weights of neurons are bounded by a polynomial in the input size, can be implemented by $\TCz$ circuits.

In this work we address this question and give evidence that efficient quantum
algorithms \emph{do not exist} for learning $\TCz$. More
concretely, under the assumption that the Ring-Learning with Errors problem
($\RLWE$)~\cite{LPR10}  cannot be solved in (quasi-)polynomial time by quantum
computers, the class of $\TCz$ functions cannot be weak-learned in quantum
(quasi-)polynomial
time. Also, under the less-standard assumption that $\RLWE$ cannot be solved in \emph{strongly sub-exponential time}, by quantum computers, we show that there exists a constant $\nu > 0$ such that there is no $n^{O(\log^\nu n)}$-time quantum learning algorithm for $\ACz$. Here, the ``strongly sub-exponential" time hardness for $\RLWE$ is defined as follows: for every constant $\eta > 2$, the $d$-dimensional $\RLWE$ problem cannot be solved in $2^{O(d^{1/\eta})}$-time by quantum computers.
$\RLWE$ is one of the leading candidates for efficient post-quantum cryptography and it is believed to be hard to solve,
even for quantum computers. For instance, the current best-known {quantum algorithms for it}~\cite{Sta19} {do not, asymptotically, perform any better than the lattice reduction techniques used for the Learning with Errors problem}~\cite{CN11,LLL82}.
\vspace{-0.8em}

\subsection{Learning models}
\label{sec:intro-learning}
In this section, we first recall the definitions of the classical and quantum learning models. For a detailed introduction to these models, we refer the reader to \Cref{sec:learningmodels}.

\vspace{-1em}
\paragraph{Classical learning model.} In the Probably Approximately Correct
(PAC) model of learning proposed by Valiant~\cite{valiant:paclearning}, a \emph{concept class} $\Cc$ is a subset of Boolean functions, i.e., $\Cc\subseteq \{c:\01^n\rightarrow \01\}$ and every
element $c:\01^n\rightarrow \01$ in $\Cc$ is referred to as a \emph{concept} . The goal of
the learning algorithm $\A$ is to learn an unknown \emph{target concept} $c\in
\Cc$ given \emph{labeled examples} of the form $(x,c(x))$ where $x$ is drawn
from an \emph{unknown} distribution $D:\01^n\rightarrow [0,1]$.  We say  that a learning
algorithm \emph{$\A$ learns $\Cc$}  if it satisfies the following: for every $c\in \Cc$
and every distribution $D:\01^n\rightarrow [0,1]$, with probability at least~$2/3$, $\A$ outputs a hypothesis $h:\01^n\rightarrow \01$ that satisfies
$\Pr_{x\sim D}[h(x)= c(x)]\geq
1-\eps$. 
{The advantage of the learner over a random guess is
$\beta := \frac{1}{2} - \frac{\eps}{2}$
and  $\beta$ is called the \emph{bias} of the learner. The learner
$\A$ \emph{properly} learns the concept class $\Cc$ if its output hypothesis is
always in the concept class, i.e., $h \in \Cc$. Otherwise, $\A$ is an improper
learner for $\Cc$.} Similarly, $\A$ is a \emph{weak learner} if $\beta=n^{-c}$ for some constant $c>0$ {and $\A$ is a \emph{strong learner} if $\beta \geq 1/6$}. In this paper all our lower bounds will be for \emph{improper
learners} {and, unless explicitly stated, for weak learners}. 

The sample complexity of a learning algorithm $\A$ is the worst-case number of
labeled examples it uses and the time complexity of $\A$ is the worst-case
running time (where the worst-case is taken with respect to the hardest concept
$c\in \Cc$ and distribution $D$). The sample/time complexity of a concept class $ \Cc$ is the
sample/time complexity  of the most efficient learning algorithm for $\Cc$. 

 In this work, we also consider learning models that relax the PAC learning framework in two ways. First we allow the learner to make \emph{membership queries} to a target concept $c$, i.e., $\A$ is allowed to ask ``what is $c(x)$" for an arbitrary $x$ of its choice. Second, instead of the learner succeeding under \emph{every} distribution $D$,  we consider the learnability of $\Cc$ when $D$ is fixed and known to the learner. We say that $\A$ PAC-\emph{learns} $\Cc$ under $D$ with membership queries if: for every $c\in \Cc$, with probability~$\geq 2/3$,~$\A$, takes labeled examples and makes membership queries, and outputs a hypothesis $h$ such that $\Pr_{x\sim D}[h(x)\neq c(x)]\leq\eps$. 
 \vspace{-1em}
\paragraph{Quantum learning model.} Bshouty and Jackson~\cite{bshouty:quantumpac} introduced the model of quantum-PAC  learning, which naturally generalizes the classical PAC model. Here, a {\em quantum} learning algorithm $\A$ has to learn the unknown concept $c\in
\Cc$ given \emph{quantum} examples of the form 
$$
\sum_x \sqrt{D(x)}\ket{x,c(x)},
$$ 
where $D :\01^n\rightarrow [0,1]$ is an unknown distribution. The goal of
the quantum learner and the notion of sample and time complexities are analogous to
the classical model.

As described in the classical setting, we also consider the model where the quantum learning algorithm is given access to an oracle	$O_c:\ket{x,b}\rightarrow \ket{x,b\oplus c(x)}$ that allows it to make \emph{quantum membership queries}. Additionally, instead of requiring the learner to succeed for \emph{all} distributions~$D$, we also consider quantum learners that learn $\Cc$ under a fixed distribution $D$. For quantum learning algorithms, a natural choice of $D$ is the uniform distribution over $\01^n$, and in this case a quantum example is given by $\frac{1}{\sqrt{2^n}}\sum_x\ket{x,c(x)}$.  We discuss the utility of such examples in the next~section. We say a learner $\A$ \emph{uniform quantum-PAC learns $\Cc$} with membership queries if: for every $c\in \Cc$, with probability $\geq 2/3$, $\A$ uses uniform quantum examples and quantum membership queries to output a hypothesis $h$ such that $\Pr_{x\sim D}[h(x)\neq c(x)]\leq\eps$.  
\vspace{-0.8em}

\subsection{Strengths of quantum examples under the uniform distribution}
\label{sec:poweroffouriersampling}

One of the main tools in quantum learning theory, when considering learning
under the uniform distribution, is the ability to efficiently perform {\em Fourier sampling}. In order to explain it, we first introduce the following. For a Boolean function $c:\01^n\rightarrow \pmset{}$, the Fourier coefficients of $c$ are given by $\widehat{c}(S)=\frac{1}{2^n}\sum_S c(x)
(-1)^{x\cdot S}$. The \emph{Fourier distribution} $\{\widehat{c}(S)^2:S\in \01^n\}$ is given by the squared Fourier coefficients of $c$ and they satisfy $\sum_S \widehat{c}(S)^2=1$. Fourier sampling refers to the task of \emph{sampling} from the Fourier distribution $\{\widehat{c}(S)^2\}_S$.

An advantage of having uniform quantum examples is that, using standard ideas from quantum information theory, a quantum learner can efficiently perform the operation
$$
\frac{1}{\sqrt{2^n}}\sum_x\ket{x,c(x)}\rightarrow \sum_S \widehat{c}(S)\ket{S}
$$
given $O(1)$ copies of $\frac{1}{\sqrt{2^n}}\sum_x\ket{x,c(x)}$. 
Measuring the resulting state allows a quantum learner to obtain a sample from the Fourier distribution.  Hence, using uniform quantum examples, one can sample from the Fourier distribution $\{\widehat{c}(S)^2\}_S$.
Classically, we do not know how to perform this sampling process (or even approximately sample) efficiently, since the Fourier coefficients $\widehat{c}(S)$ depends on $2^n$ values of~$x$.	Therefore, one avenue to obtain quantum speedups with uniform quantum examples arises from the use of Fourier sampling.

Indeed, quantum Fourier sampling has been profitably used in many applications.
This idea was first used in the aforementioned paper of Bernstein and Vazirani~\cite{bernstein&vazirani:qcomplexity}, where they observe that the Fourier support of linear functions is concentrated on a single point. Therefore, unknown linear functions can be learned using just one uniform quantum example. Fourier sampling was later used by Bshouty and Jackson~\cite{bshouty:quantumpac} to show that $\DNF$s can be learned quantum-{efficiently}. A classical analogue of this question is a long-standing open question.\footnote{With 
classical membership queries, $\DNF$
formulas can be learned in classical polynomial time~\cite{jackson:dnf}.} Kanade et al.~\cite{kanade:dnfs} extended the result of Bshouty and Jackson and showed how to learn $\DNF$s quantum-efficiently even
under product~distributions. 

In fact, a notorious bottleneck in classical learning theory is obtaining a polynomial time learning algorithm for the class of  $O(\log n)$-juntas\footnote{A $k$-junta is a Boolean function on $n$ bits, whose output only depends on $k$ out of the $n$ input bits.} which are a subset of polynomial-sized $\DNF$s. By contrast, At\i c\i\ and
Servedio~\cite{atici&servedio:testing} showed that $O(\log n)$-juntas can be learned quantum-efficiently under the uniform distribution. In fact the time-efficient
learnability of $O(\log n)$-juntas can be used to time-efficiently quantum learn the concept class $\NCz$ under the uniform distribution (for a proof of this, see \Cref{sec:nczeasy}).\footnote{$\NCz$ is the concept class of constant-depth circuits consisting of fan-in $2$ $\AND, \OR, \NOT$ gates.} Subsequently, Arunachalam et al.~\cite{arunachalam:ksparse} showed that the concept class of $k$-Fourier-sparse Boolean functions (which includes both $(\log k)$-juntas and linear functions) can be learned query-efficiently given quantum examples. 

One thing common to all these results was the application of the Fourier sampling to learn a concept class under the uniform distribution, which was key for the large quantum speedup. 
\vspace{-0.8em}

\paragraph{Related work.} In the context of learning, apart from the uniform distribution setting, some works have focused on understanding the power of quantum learning under \emph{arbitrary distributions}~\cite{servedio&gortler:equivalencequantumclassical,aticiservedio:pac,arunachalam:paclearning}. In particular, \cite{arunachalam:paclearning} showed that quantum examples do not provide an advantage over classical examples for PAC learning. 

Recently, both shallow circuits and the Learning with Errors problem have been used in different contexts to understand the capabilities of quantum computation. Grilo et al.~\cite{grilo:LWEeasy} showed that polynomially many \emph{quantum examples} suffice to solve the $\LWE$ problem in quantum polynomial time, while this remains a hard problem when given just \emph{classical} examples.  Bravyi et al.~\cite{bravyi2018quantum} exhibited a problem which can be solved using shallow quantum circuits but requires logarithmic depth classical circuits (with bounded fan-in) to solve. Bene Watts et al.~\cite{WKST19} improved their result by removing the ``bounded fan-in" assumption. In another context, under the assumption that the Learning with Errors problem is hard for quantum computers (given just classical examples),  Mahadev~\cite{mahadev:classicalverification} demonstrated a classical protocol to classically verify the result of an efficient quantum computation. 

\subsection{Our results}
In this paper we address two natural questions that arise from the work of Bshouty and Jackson~\cite{bshouty:quantumpac}, which showed the quantum-efficient learnability of depth-$2$ circuits.

The first question is, can the work of~\cite{bshouty:quantumpac} be extended 
to learn depth-$3$ circuits, or more generally, constant-depth polynomial-sized circuits (i.e., $\TCz$ and $\ACz$) in quantum polynomial time? Classically, in a seminal result, Linial, Mansour and Nisan~\cite{LMN93} constructed an $n^{O(\polylog n)}$-time learning algorithm for $\ACz$ by approximately learning the Fourier spectrum of an unknown $\ACz$ circuit. Subsequently,
Kharitonov~\cite{Kharitonov93} showed that their learning algorithm is optimal assuming that
factoring cannot be solved in quasi-polynomial time. Since factoring can be
solved in \emph{quantum} polynomial time with Shor's algorithm~\cite{Shor}, the lower
bound of Kharitonov doesn't apply to quantum learners. Moreover, since Fourier
sampling is easy quantumly (under the uniform distribution), it seems
plausible that one could efficiently learn important properties of the Fourier spectrum
of $\ACz$ circuits (similar to the work of Linial et al.~\cite{LMN93}). This
could possibly result in efficient quantum learning algorithms for $\ACz$ under
the uniform distribution (similar to many results discussed in the previous
section).
  
The second question is, can the work of~\cite{bshouty:quantumpac} be extended to
the class of depth-$2$ \emph{threshold} circuits (known as $\TCz_2$), i.e., can
$\TCz_2$ be learned quantum-efficiently? We notice that $\TCz$ circuits are {\em practically} very relevant  since
constant-depth polynomial-sized feed-forward neural networks with weights (of the neurons) bounded by some polynomial in the input size, can be implemented as circuits in $\TCz$~\cite{maas:sigmoidboolean,goldmann:thresholdneural}. If there were efficient quantum algorithms for learning $\TCz$, then it is plausible that quantum computers could give an enormous advantage in approximately learning the weights for neural networks.

In this paper, we give a conditional negative answer to both questions. In particular, we show that under the assumptions that support the security of current post-quantum
cryptography: $\TCz$ and $\ACz$ cannot be learned efficiently
by quantum computers under the uniform distribution (under different
assumptions); and 
$\TCz_2$ cannot be
PAC learned efficiently by quantum computers. We summarize these results in the table below.\footnote{$\mathsf{BQP}$, or \emph{bounded-error quantum polynomial time}, is the class of decision-problems that be solved in polynomial time on a quantum computer with bounded-error.}\\ %
\renewcommand\arraystretch{2}
\begin{table}[!ht]
	\scriptsize
	\centering
	\begin{tabular}{|>{\centering\arraybackslash}m{1.4in}|>{\centering\arraybackslash}b{1.1in}|>{\centering\arraybackslash}m{0.8in}|>{\centering\arraybackslash}m{1.6in}|@{}m{0pt}@{}}
		\hline
		\Centerstack{No learner in this model} & {Running in} & \Centerstack{For the \\ complexity class} & Assuming    \\\hline
		\multirow{1}{*}[-2em]{ \Centerstack{Uniform-distribution \\PAC\\ (with
    membership queries)} }&\Centerstack{$n^{O(\log^\nu n)}$ time for a\\ constant $\nu > 0$}&\Centerstack{$\ACz$}&\Centerstack{No strongly sub-exponential time \\ algorithm for $\RLWE$} \\
    	\cline{2-4}& {$2^{O(\polylog n)}$ time} &$\TCz$&\Centerstack{No quasi-polynomial time \\ %
    algorithm for $\RLWE$} \\
		\cline{2-4}& $\poly(n)$ time &$\TCz$&$\RLWE\notin \mathsf{BQP}$\\\hline %
		Distribution-free PAC& {$\poly(n)$ time} & $\TCz_2$  &  $\LWE\notin \mathsf{BQP}$  \\\hline %
	\end{tabular}
	\caption{\small Hardness of learning results for $\poly(n)$-sized circuits.}
	\label{tab:res}
\end{table}
These results give a strong negative answer to a question of
Aaronson~\cite{Aaronson-blog}, under cryptographic assumptions. Aaronson asked if, $\TCz$ and $\ACz$ can
	be \emph{quantum-PAC learned} in polynomial time. Our first result already
  gives a (conditional)
	negative answer even  when we fix the \emph{uniform distribution} and
  allow the learner to make quantum membership queries. Our second result gives a
	conditional refutation to the PAC learnability of $\TCz$ even when restricted to
	depth-$2$ threshold~circuits.

In order to achieve our results, we follow a strategy proposed by
Valiant~\cite{valiant:paclearning}, who showed the hardness of \emph{proper} learning  
the class of polynomial-sized circuits based on the security of
cryptographic objects. This strategy was subsequently improved upon to give
conditional \emph{improper} learning lower bounds and used to prove the classical hardness of various concept classes~\cite{Kharitonov:AC1,Kharitonov93,KearnsV94,KlivansS09}.
These results have the following structure: assuming there exists an efficient learning
algorithm for some concept class $\Cc$, there exists an adversary that is able
to break some cryptographic construction using the learning algorithm as a subroutine. Here, the adversary provides the resources to the learning algorithm based on the cryptographic primitive it is trying to break. This implies that if the cryptographic construction is secure, then such a  learning algorithm cannot exist.

In this paper, we quantize these well-studied classical proof-of-hardness
techniques. The difficulties in quantizing such results are three-fold. First, many of the classical hardness of learning results rely on cryptographic
primitives whose security is based on the hardness of factoring. As stated
previously, this would not hold in the quantum regime due to Shor's quantum
polynomial-time factoring algorithm~\cite{Shor}. {Second, the fact that adversaries can
efficiently create classical examples from some distribution $D$ does not imply
that {\em quantum examples} can be created according to the same distribution.
An important issue we run into in this case is solving the index-erasure problem, which is known to be hard to solve on quantum computers~\cite{ambainis:indexerasure,rosmansis:index}.
Finally, some of the hardness results implicitly use the fact that the learning
algorithm they are considering is classical and the proof techniques do not
follow through in the quantum setting. For example,
Kharitonov~\cite{Kharitonov93} uses collision arguments to bound the amount of
information retrieved by the learner, but this approach does not work quantumly. We discuss these issues in further detail in the next section.}

In subsequent sections, we delineate the connections between the hardness of
quantum learning and the security of certain cryptographic
primitives -- specifically, quantum-secure {family of pseudo-random
functions} and
public-key encryption schemes. Next, we sketch how to use these connections to
show hardness of quantum learning for some interesting concept classes. 

\vspace{-1em}

\subsubsection{Pseudo-random functions vs. quantum learning}
A {family of pseudo-random functions (PRF) are cryptographic objects}
that ``mimic" random functions.
More concretely, a PRF is a keyed-family of functions $\Fe = \{f_{\mathbf{k}} : \01^n \to
\01^{\ell} \text{ where } {\mathbf{k}} \in \01^M \}$ such that
no efficient adversary that is given query access to $f$ can distinguish
if $f$ was uniformly picked from $\Fe$ or a truly random function with non-negligible advantage
over a random guess.\footnote{By \emph{negligible advantage}, we mean: for every
$c>0$ the advantage is at most $1/n^c$.} 
If even polynomial time \emph{quantum} adversaries cannot distinguish such
cases with quantum query access to the function, then we say that the PRF is \emph{quantum-secure}.
We informally state our first result below (see
\Cref{cor:quantumsafePRF} for a full statement).

\begin{result}
  \label{result:prg}
  If there is a quantum-secure PRF $\Fe$,  then $\Fe$ does not have an efficient uniform quantum-PAC learner with membership~queries.
\end{result}

Kharitonov~\cite{Kharitonov93} established the connection between Pseudo-random
generators (PRGs) and classical learning by constructing a circuit class such
that the PRG is computed by the circuit class. He proceeded to show that if
a learning algorithm for such a concept class existed, then would be possible to
break the PRG. {This approach (implicitly) requires that while the
complexity of the PRG scales polynomially with the size of the seed, it achieves a super-polynomial stretch}
Such a requirement is satisfied, for instance, by the BBS PRG~\cite{BBSPRG}, whose security relies on hardness of
factoring.

Unfortunately, no post-quantum PRGs are known with this property.
Inspired by the ideas used in~\cite{Kharitonov93}, we overcome this by
considering the  connection between quantum-secure PRFs and we
prove Result~\ref{result:prg}.
We stress that the proof in~\cite{Kharitonov93} does not trivially quantize.
For instance, 
the crux in Kharitonov's proof is
determining the probability of collision in the classical examples, and this does
not make sense in the quantum regime. Clearly, each quantum example contains
information about every input. However, we can show that efficiently accessing this information
simultaneously for most $x$ is information-theoretically impossible. 

We now sketch a proof of \Cref{result:prg}. Let $\Fe = \{f_{\mathbf{k}} : \01^n \to \01
\text{ where } {\mathbf{k}} \in \01^M\}$.
By contradiction, let us assume there exists an efficient uniform-PAC quantum
learning algorithm for $\Fe$. Using this efficient learner, we show how to construct a
quantum distinguisher for the PRF $\Fe$. As $\Fe$ is assumed to be
quantum-secure, we obtain a contradiction and prove the result. More
concretely, let $f$ be a a function and suppose a distinguisher with access to
$f$ has to determine if it comes from $\Fe$ or it is a truly random function.
Consider a distinguisher that simulates a quantum learning algorithm $\A$ by
using his oracle to answer the learner's queries. %
$\A$ is supposed to output a hypothesis $h$ that approximates $f$. The
distinguisher then picks a random challenge $x^* \in \01^n$ and outputs $1$ if
and only if $h(x^*) = f(x^*)$. One technical aspect that we show here is that, if the
learning algorithm has bias $\biaslearning(n)$, then the distinguisher has an advantage $\geq \frac{\biaslearning(n)}{2}$ in 
distinguishing if $f$ comes from the pseudo-random function family or is a random
function. In particular, if
there is a polynomial time weak quantum learning algorithm for $\Fe$ under the uniform distribution, then there exists a polynomial-time quantum
distinguisher $\De$ and some constant $c>0$ that satisfies
\begin{align}
\label{eq:distinguisherinintro}
  \Big \vert \Pr_{f \in \Fe} [\De^{\ket{f}}(\cdot)]- \Pr_{f \in
  \mathcal{U}}[\De^{\ket{f}}(\cdot)] \Big\vert \geq \frac{1}{n^c},
\end{align}
where, in the first case, $f$ is picked uniformly at random from $\F$ and, in the
second case, $f$ is a uniformly random function. The key point in order to prove that a quantum learner with advantage
$\biaslearning(n)$ implies a distinguisher with advantage
$\frac{\biaslearning(n)}{2}$ is to bound the success probability of learning a
random function.

\subsubsection{Public-key encryption schemes vs. quantum learning}

A public-key encryption scheme consists of a triple of algorithms
(Key-generator, $\enc$, $\dec$). Key-generator is a randomized
algorithm that on input $1^n$ (where $n$ is a security parameter)
outputs a tuple $(\Kpub, \Kpriv)$, where $\Kpub$ is a publicly known key used to encrypt messages and $\Kpriv$ is a private key used to decrypt messages.
$\enc$ is a deterministic  algorithm that receives as input the public key $\Kpub$, some randomness $r$ and
a message $b\in \01$ and outputs $\enc(\Kpub, r,b)$, which we denote by
$\enc_{\Kpub}(r,b)$ for simplicity. $\dec$ receives as input the private key
$\Kpriv$ and $\enc_{\Kpub}(r^*,b^*)$ and outputs $c\in \01$ (we write $\dec_{\Kpriv}$  in order be explicit about the dependence of $\dec$ on $\Kpriv$). The
public-key encryption scheme is said to be correct if: for uniformly random values $r^*$ and~$b^*$
$\dec_{\Kpriv}\left(\enc_{\Kpub}(r^*,
b^*)\right) \ne b^*$ with negligible probability. 
An encryption scheme is (quantum) secure if,
given $\Kpub$ and $\enc_{\Kpub}(r^*,b^*)$, a (quantum) polynomial time adversary
can output $b^*$ with at most a negligible advantage over a random
guess.

We connect quantum-secure public-key encryption schemes to the
hardness of learning as follows (see \Cref{thm:cryptotolearn} for a full
statement). %
\begin{result}
  \label{result:encryption}
Let  $\mathsf{S}$ be a quantum-secure public-key cryptosystem. If $\Cc_{\mathsf{S}}$ is the concept class containing the decryption functions of the cryptosystem $\mathsf{S}$, then there is no efficient weak quantum-PAC learner for~$\Cc_{\mathsf{S}}$.%
\end{result}
The works of Kearns and Valiant~\cite{KearnsV94} and Klivans and
Sherstov~\cite{KlivansS09} provide a connection between public-key encryption schemes and
learning functions. They showed that if there exists a PAC learning algorithm
for a concept class that contains the decryption function $\dec_{\Kpriv}$, then it is
possible to predict $b^*$ from $\Kpub$ and $\enc_{\Kpub}(r^*,b^*)$ with a 
$1/\poly(n)$ bias. They prove this by simulating the learning
algorithm as follows: the distinguisher prepares examples of the form $(r,\enc_{\Kpub}(r,b))$ for (uniformly) random $r$ and $b$. Using the guarantees of the classical PAC learning algorithm and the correctness of the encryption scheme, they show that the hypothesis~$h$ output by the learner satisfies 
$$
\Pr_{r^*, b^*}\Big[h\Big(\enc_{\Kpub}(r^*,b^*)\Big)=b^*\Big] \geq
\frac{1}{2}+\frac{1}{n^c},
$$
for some $c>0$. 
In this paper, we quantize their argument, but the situation is much more intricate than in the classical case. Classically, $r$ and $b$ can be picked uniformly at random
at each step in order to create a new training example $(\enc_{\Kpub}(r,b),b)$.
Quantumly, however, we do not know of an efficient way to create a quantum example
$\frac{1}{\sqrt{2|\mathcal{R}|}}\sum_{r,b} \ket{\enc_{\Kpub}(r,b)}\ket{b} $, where
$\mathcal{R}$ is the space of the possible randomness. Notice that a
straightforward way of preparing this state involves solving the index-erasure
problem~\cite{ambainis:indexerasure,rosmansis:index}, which is
conjectured to be a hard problem to solve on a quantum computer. See \Cref{sec:quantum-safelearning} for more
details. 

{
	Instead, we first define a distribution $D$ as follows: pick $\poly(n)$-many
  uniformly random $(r,b)$ and let $D$ be the uniform distribution over
  $\{\enc_{\Kpub}(r,b)\}$ where the set ranges over the $\poly(n)$-many observed
  $(r,b)$. Our hope is to run the quantum learner on this distribution $D$
  (which we are allowed to since we assumed that it is a quantum-PAC learner).
  However, we run into an issue which~\cite{KlivansS09} need not worry about.
  Let $\enc_{\Kpub}(r^*,b^*)$ be the challenge string that the distinguisher
  needs to correctly decrypt to $b^*$. Observe that $\enc_{\Kpub}(r^*,b^*)$ need
  not even lie in the support of $D$, so running a quantum learner on the
  distribution $D$ might not even help the distinguisher in predicting $b^*$. In
  contrast, in the simulation argument of Klivans and Sherstov~\cite{KlivansS09}
  the pair $(\enc_{\Kpub}(r^*,b^*),b^*)$ is always in the support of the distribution, since $r$ and~$b$ are picked uniformly at random to create the classical example  $(\enc_{\Kpub}(r,b),b)$.

	 Ideally, we would like to use our PAC learner on a distribution $D'$ for which
$(\enc_{\Kpub}(r^*,b^*),b^*)$ is the support of $D'$. This would enable the distinguisher to use the guarantees of a quantum learner when run on $D'$.  The challenge here is that the distinguisher would need to find such a~$D'$ without prior knowledge of $\enc_{\Kpub}(r^*,b^*)$ and $b^*$!
}%
We circumvent this issue with the following observation: if two distributions are sufficiently close to each other, then the learner should perform ``essentially equivalently" on both distributions. In particular, we use a training distribution $D$ that is close enough to the testing distribution $D'$ containing $(\enc_{\Kpub}(r^*,b^*),b^*)$ so that the learning algorithm is unable to distinguish them. 

We now provide more details here. Suppose we have a quantum-secure public-key cryptosystem
with a (randomized) encryption function $\mathsf{Enc}_{\Kpub,r}:\01\rightarrow
\01^n$ and decryption function $\mathsf{Dec}_{\Kpriv}: \mathcal{E}\rightarrow
\01^n$, where $\mathcal{E}$ is the set of all valid encryptions and
$(\Kpub,\Kpriv)$ is the output of the Key-generation algorithm. Assume that the challenge string is  $\enc_{\Kpub}(r^*,b^*)$ for a uniformly random $r^*,b^*$ and an adversary for this cryptosystem has to correctly guess $b^*$.

In order to construct such an adversary, first define the concept class $\Cc=\{\mathsf{Dec}_{\Kpriv}:\Kpriv\}$, the set of
all decryption functions (one for each private key). Furthermore, assume that
there is a weak quantum-PAC learner for $\Cc$ that uses $L$ examples, i.e., a polynomial-time quantum learning algorithm $\A$ which receives
$L$ quantum examples $\sum_{x}\sqrt{D(x)}\ket{x,c(x)}$ (for an unknown
distribution $D$ and concept $c\in \Cc$) and outputs a hypothesis $h$ that is close
to the concept $c$.

As discussed previously, we now define a meaningful distribution $D$ on which the distinguisher runs the quantum learner: consider a set $S$ with $L^3$ tuples $(r_i,b_i)$ that are chosen uniformly at random from their respective domains; let $D$ be the uniform distribution over  $\{\enc_{\Kpub}(r,b):(r,b)\in S\}$. The adversary behaves as follows: run the quantum-PAC learner $\A$ under the distribution $D$ and provide it $L$ quantum examples of the form
$$ \ket{\psi}=\frac{1}{\sqrt{L^3}} \sum_{(r,b) \in S} \ket{\enc_{\Kpub}(r,b)}\ket{b}. $$
When $\A$ outputs the hypothesis $h$, the adversary outputs $h\big(\enc_{\Kpub}(r^*,b^*)\big)$ as its guess for $b^*$.
Notice that with overwhelming probability, $S$ does not contain the tuple $(r^*,
b^*)$ corresponding to the challenge. So there should be no guarantee on  the value of $h\big(\enc_{\Kpub}(r^*,b^*)\big)$. In order to overcome this, consider a quantum example state $$ 
\ket{\psi'}=\frac{1}{\sqrt{L^3 + 1}} \left(\ket{\enc_{\Kpub}(r^*,b^*)}\ket{b^*} + \sum_{(r,b)
\in S} \ket{\enc_{\Kpub}(r,b)}\ket{b}\right).
$$
We show that as the learner uses only $L$ quantum examples, $\ket{\psi}^{\otimes L}$, the output statistics of every quantum learning algorithm, when run on $\ket{\psi}$ and $\ket{\psi'}$, is very similar. In fact, we show that the distribution on the hypothesis set is almost the same in each case.  Now, using the performance guarantees of a quantum-PAC learning algorithm and the closeness of the distributions between the hypothesis sets, we conclude that $h\big(\enc_{\Kpub}(r^*,b^*)\big)$ equals $b^*$ with probability at least $\frac{1}{2}+\frac{1}{\poly(n)}$. This contradicts the quantum-secure assumption on the cryptosystem. 
\vspace{-0.8em}

\subsubsection{Conditional hardness of learning \texorpdfstring{$\TCz$}{TC0} and \texorpdfstring{$\ACz$}{AC0}} 

\paragraph{Hardness of $\TCz$.} 
The first consequence of~\Cref{result:prg} is to
give a strong \emph{negative answer} to the question of
Aaronson~\cite{Aaronson-blog} regarding the quantum learnability of $\TCz$
circuits, under cryptographic assumptions. 
\begin{result}
  \label{result:tc0}
  If the Ring-Learning with Errors problem cannot be solved in quantum polynomial time, then there is no polynomial-time uniform weak quantum-PAC learner for $\TCz$ with membership~queries.
\end{result}

As previously mentioned, the Ring-Learning with Errors problem -- i.e., the Learning with Errors problem defined over polynomial rings instead of matrices -- is a
cornerstone of state of the art efficient post-quantum cryptosystems and is
widely believed to be hard for quantum computers. %
The starting point for proving \Cref{result:tc0} is a ring variant of the pseudo-random function
family
presented by Banerjee, Peikert and Rosen~\cite{BPR12}, which was proven to be
quantum secure by Zhandry~\cite{zhandry:qprf}. {We show that our variant of the PRF 
is quantum secure under the assumption that the \emph{Ring-Learning with Errors} ($\RLWE$)
problem cannot be solved by quantum computers efficiently. We do not define the
$\RLWE$ problem here (see \Cref{sec:rlwe} for details), but point out that
$\RLWE$ is believed to be as hard as $\LWE$.
Our variant of the PRF, denoted $\RF$, satisfies two crucial 
properties: }
(a) every $f \in \RF$ can be computed by a $\TCz$ circuit; and
(b) suppose there exists a
\emph{distinguisher} $\De$ for $\RF$ -- i.e., there exists some constant $c > 0$ and a polynomial-time algorithm
$\De$ that satisfies
\begin{align}
  \Big \vert \Pr_{f\in \RF} [\De^{\ket{f}}(\cdot)=1]-
  \Pr_{f\in\mathcal{U}}[\De^{\ket{f}}(\cdot)=1] \Big\vert \geq \frac{1}{n^c} 
\end{align}

\noindent where $f\in\mathcal{U}$ denotes a uniformly random function -- 
then there exists a polynomial-time algorithm that solves $\RLWE$.
We can then use our connection between
PRFs and the hardness of quantum learning (\Cref{result:prg}) to prove \Cref{result:tc0}.

We remark that this work contains the first {explicit} proof that such a PRF
can indeed be implemented in $\TCz$. Building on results laid out
in~\cite{LPR13, HAB02} for ring arithmetic, this is shown by carefully
considering the cost of efficiently representing ring elements and performing
operations such as iterated multiplication over these ring elements.
Also, note that it is crucial to consider the $\RLWE$ version of the PRF instead of
the standard $\LWE$ version as the matrix arithmetic operations needed to implement the
latter would require log-depth circuits. 

Given the translation between feed-forward neural networks and $\TCz$ circuits~\cite{maas:sigmoidboolean}, our
negative result on learning $\TCz$ implies that, quantum
resources do {not give an exponential} advantage in learning the weights
of such neural networks (assuming $\LWE$ is quantum-hard).

We now make an alternate choice of parameters for the PRF $\RF$ in order to prove a stronger hardness result on quantum learning $\TCz$. However, a caveat of using alternate parameters is, we need to make stronger assumptions on the hardness of the $\RLWE$ problem. 

\begin{result}
	\label{result:strongertc0}
	If the Ring-Learning with Errors problem cannot be solved in quantum quasi-polynomial time, then there is no quasi-polynomial time uniform weak quantum-PAC learner for $\TCz$ with membership~queries.
\end{result}

As far as we are aware, the above result wasn't known even in the \emph{classical} literature (the only result in this flavor was proven by Fortnow and Klivans~\cite{Fortnowklivans:hardness} who proved the hardness of learning based on circuit complexity theoretic assumptions). Kharitonov's proof with the BBS PRG~\cite{BBSPRG} cannot yield the above result. In fact, allowing a quasi-polynomial time learner for the BBS PRG, using Kharitonov's proof, will result in a \emph{weak}-distinguisher, i.e., an algorithm that can distinguish between a random function and a PRF with $2^{-\polylog n}$ probability. By contrast, our Result~\ref{result:prg} allows us to improve upon Kharitonov's result by showing that a quasi-polynomial time learner yields a distinguisher with $1/\poly(n)$ probability of distinguishing between a random function and a function from $\RF$.

\paragraph{Hardness of $\ACz$.}
Here we show how to choose a new set of parameters for the
PRF $\RF$ that allows it to be implemented
in $\ACz$. Unfortunately, this comes with the added costs of requiring stronger
assumptions on the hardness of $\RLWE$ and weakening the achieved conclusion,
namely our hardness result holds for {\em
strong} learners (i.e., learners with bias $\beta \geq 1/6$). We achieve the following result.

\begin{result}
  \label{result:ac0}
  If there is no strongly sub-exponential time quantum algorithm for the Ring-Learning with Errors problem, then there is no $n^{O(\log^{\nu} n)}$-time uniform {strong} quantum-PAC learner for $\ACz$ circuits on $n$ bits (for a {constant $\nu > 0$}), using membership~queries, where $d = O(\polylog n)$.
\end{result}
\vspace{-0.8em}

In order to better understand this assumption, we remark that the current best known classical algorithms for $\LWE$ require exponential
time~\cite{LLL82,CN11} and any straightforward quantization of
these results gives only a polynomial speedup. These algorithms are currently believed to be the best
known algorithms for $\RLWE$ as well. More specifically, for $d$-dimensional $\RLWE$ instances of interest in~\Cref{sec:AC0prghardness}, these algorithms would scale as $2^{O(\sqrt{d})}$-time. Therefore, any non-trivial improvements in the exponent for these $\RLWE$ algorithms would also be a breakthrough for
the algorithms and cryptanalysis communities.

There is a key difference involved in proving \Cref{result:ac0}, as compared to \Cref{result:tc0}. An inherent problem in directly using the $\RLWR$-based PRF $\RF$ to prove \Cref{result:ac0} is that, any circuit computing a function $f \in \RF$ needs to compute an iterated multiplication of many ring elements. While this can be done efficiently with a $\TCz$ circuit, it cannot be computed using an $\ACz$ circuit. In fact, it is \emph{a priori} unclear if we can devise PRFs with {super-polynomially large truth tables} whose hardness is based on a quantum-hard problem while also being computable in $\ACz$.\footnote{Here we mean that the size of the truth table for each function in the PRF is super-polynomial in the dimension of the underlying quantum problem.} 

In order to overcome this issue, we use a technique that was also used by Kharitonov ~\cite{Kharitonov93}; by considering PRFs with a truth-table size that is sub-exponential in the dimension of the underlying $\RLWR$ problem. Specifically, we
obtain a PRF $\RF'$ whose truth-table size is polynomial in $n$ while the dimension, $d$, is poly-logarithmic in $n$. We show that computing an arbitrary function $f \in \RF'$ can be performed by an $\ACz$ circuit.
However, this increase in truth table size -- from being super-polynomial in $d$ as in the $\TCz$ case to being sub-exponential in $d$ for $\ACz$ -- which means that a stronger hardness assumption is needed to guarantee the security of this PRF. To this end, by assuming that 
$d$-dimensional $\RLWE$ cannot be solved by quantum computers in strongly sub-exponential time (i.e., $2^{O(d^{1/\eta})}$-time for every constant $\eta > 2$), the PRF that we consider here is quantum-secure. With this new PRF such that $d = O(\polylog n)$, we repeat the arguments of \Cref{result:tc0} and show the $n^{O(\log^\nu n)}$-time hardness of quantum learning for $\ACz$ for a constant $\nu > 0$ that depends on $\eta$.

We remark that Result~\ref{result:ac0}  matches the upper bound of Linial,
Mansour and Nisan~\cite{LMN93} (up to the constant
$\nu$ in the exponent), who showed how to \emph{classically} learn $\ACz$
circuits on $n$ bits using $n^{O(\polylog(n))}$ many random examples $(x,c(x))$
where $x$ is drawn from the uniform distribution. Our result shows that, up to
the $\nu$ in the exponent, the classical learning algorithm of~\cite{LMN93} is
optimal even if one is given access to uniform \emph{quantum examples} and is allowed to make quantum queries, under the $\RLWE$ assumption explained previously. 

It is worth noting that our result also implies conditional hardness results for
{\em classical} learning of $\TCz$ and $\ACz$, under the
assumption that solving $\RLWE$ is hard for classical computers (instead of factoring in
the case of Kharitonov~\cite{Kharitonov93}). 

\vspace{-0.5em}
\subsubsection{Conditional hardness of PAC learning \texorpdfstring{$\TCz_2$}{TC0[2]}}
Our third hardness result  is the following.
\begin{result}
  \label{result:tc02}
  If the Learning with Errors problem cannot be solved in quantum
  polynomial time, then there is no polynomial-time weak quantum-PAC learner for $\TCz_2$.
\end{result}
\vspace{-0.5em}
In contrast to our first hardness result for $\TCz$, we
stress that the quantum learners in this result are
\emph{quantum-PAC learners}. The main idea to prove this result is to consider the
$\LWE$-based public-key cryptosystem proposed by Regev~\cite{regev:lattice} (see \Cref{sec:cryptosystemLWE}). Klivans and
Sherstov~\cite{KlivansS09} considered this cryptosystem and showed that the decryption functions in this cryptosystem can be implemented by circuits in $\TCz_2$. We can then use our connection between quantum learning and quantum-secure cryptosystems (\Cref{result:encryption}) to derive \Cref{result:tc02}.\\

\vspace{-2em}

\subsection{Open questions}
This work raises a number of interesting open questions, which we list below. 
\vspace{-0.8em}

\paragraph{Learning $\ACz_d$.}
What is the smallest $d$ for which we can prove that quantum learning $\ACz_d$ is (conditionally) hard? Bshouty and Jackson~\cite{bshouty:quantumpac} gave a quantum polynomial-time algorithm for $\ACz_2$ and for some universal constant $d'\geq 3$ (independent of the input-size of the $\ACz_{d'}$ circuit), our work rules out polynomial-time learning algorithms for $\ACz_{d'}$, assuming there exists no sub-exponential time quantum algorithm for the $\RLWE$ problem. Classically, using the constant-depth construction of PRGs by Naor and Reingold~\cite{naorreingold:prf} in Kharitonov's~\cite{Kharitonov93} result, one can show that $\ACz_5$ is hard to learn (assuming factoring is hard on a classical computer). %

\paragraph{Uniform learning of $\TCz_2$.}
The conditional hardness results of $\TCz$  under the uniform
distribution and hardness of quantum-PAC learning $\TCz_2$  do not rule out the possibility that $\TCz_2$ admits polynomial-time quantum learning algorithms under the uniform distribution. It is an open question to show
if $\TCz_2$ can be learned quantum-efficiently under the uniform distribution
or if we can show a conditional hardness result.
\vspace{-0.8em}

\paragraph{Hardness of (quantum) learning quantum shallow circuits} Linial et al.~\cite{LMN93} showed the quasi-polynomial time learnability of shallow \emph{classical} circuits. Correspondingly, what would be the time/sample complexity of learning shallow \emph{quantum} circuits with a quantum learner? Recently Chung and Lin~\cite{chung:learningquantumcirc} showed that polynomially many samples suffice to learn the class of polynomial-sized \emph{quantum circuits}.%
\vspace{-0.8em}

\paragraph{Hardness of quantum learning from other assumptions} 
Finally, we leave as an open question the possibility of proving the hardness of
quantum learning from other complexity-theoretic assumptions. We now point to some
potential directions:
\begin{itemize}
  \item Recent results have proved Strong Exponential Time Hypothesis (SETH)-based hardness of
    the Gap-SVP~\cite{aggarwal:Gapsvp} problem and the 
$    \mathsf{Subset-sum}$~\cite{abboud:seth} problem (actually, they prove the limit of a
    specific approach to solve the $   \mathsf{Subset-sum}$ problem). These hardness results do
    not imply any hardness for learning problems directly and we wonder if they can be
    tightened in order to make it possible.
    
  \item Oliveira and Santhanam~\cite{oliveira:conspiracy} established a connection
between learning theory, circuit lower bounds and pseudo-randomness. Is it
    possible to quantize such connections?

 \item Daniely and Schwartz~\cite{daniely&shalevshwartz:limitdnf} showed a
   complexity-theoretic hardness of PAC learning DNFs. Could we also show
    hardness for quantum-PAC learning DNFs as well?

\end{itemize}
\vspace{-0.8em}

\subsection{Differences with previous version}

This version of this work differs from the previous
version, due to a mistake in the former proof. We discuss the main change below.

In the
previous version of this work, we claimed to prove hardness of learning $\TCz$
and $\ACz$ based on the quantum security of pseudorandom generators. {In our
proof, we assumed that the run-time of a learning algorithm depended only on the {\em input size}. 
Essentially, we considered a learning algorithm $\A$ that runs in time $n^a$ with bias ${n^{-c}}$ for a concept class on $n$ bits, and showed the existence of a PRG $G$ that takes an $n$-bit seed having a stretch $n^d$ where $d > a+c$ that $\A$ cannot learn efficiently.}

However, in the standard definition of polynomial-time learnability, the run-time of a quantum learning algorithm should depend both on the input size $n$ as well as the 
representation size of the {\em concept class}. In other words, an efficient learner $\A$ for a PRG $G$, with an $n$-bit seed, stretch $n^d$ and representation size $n^s$, is allowed to take time $O(\poly(n, n^s))$.
Under this definition, for every instantiation of our PRG, since $n^s > n^d$, there exists an efficient learning algorithm with run-time $n^a > n^s > n^d$, thereby not leading to any  contradiction, or hardness-of-learning the concept~class.

In order to achieve our hardness results in this updated version, we use
\emph{pseudorandom functions} such that for every choice of parameters, their truth-table 
size $T(n) \gg \poly(n, s(n))$, where $s(n)$ is the representation size of the function. 
Hence for every learning algorithm $\A$ running in time $t(n)$, we show the existence of 
a PRF with truth table size $T(n)\gg O(\poly(n, s(n)) \geq t(n)$ that $\A$ cannot learn efficiently.

	\subsection*{Organization} In \Cref{sec:prelim}, we state some required
  lemmas, discuss cryptographic primitives and formally define the classical and quantum
  learning models. In \Cref{sec:hardnessofPRF}, we discuss the
  (Ring)Learning with Errors problem and its variants along with the pseudo-random functions constructed from these problems. In \Cref{sec:quantumsafePRFvslearn}, we describe the connection between quantum-secure PRFs and quantum learning under the uniform distribution and prove our main results showing the hardness of learning $\TCz$ and $\ACz$. Finally, in \Cref{sec:quantum-safelearning}, we show the connection between quantum-secure public-key cryptosystems and quantum-PAC learning and conclude with the proof for the hardness of PAC-learning $\TCz_2$.
\vspace{-0.8em}

\subsection*{Acknowledgments}
S.A.~thanks Andrea Rocchetto and Ronald de Wolf for discussions during the
initial stage of this work and Robin Kothari for introducing him to this
problem. S.A.~and A.S.~did part of this work during the “Workshop on Quantum
Machine Learning” at the Joint Centre for Quantum Information and Computer
Science, University of Maryland. S.A.~did part of this work at QuSoft, CWI and
was supported by ERC Consolidator Grant QPROGRESS. We thank Anurag Anshu, Gorjan Alagic, Daniel Apon and Prabhanjan Ananth for helpful discussions. We thank Robin Kothari  and Ronald de
Wolf for many helpful comments which improved the presentation of this paper. S. A. thanks Aravind Gollakota and Daniel Liang for useful discussions. We thank Oded Regev for noticing a mistake in a previous version of this~paper. We are grateful to Leo Ducas for helping us find a suitable set of parameters for our Ring-LWE hardness assumptions.

	\section{Preliminaries}
	\label{sec:prelim}
	
	\subsection{Notation and basic claims}
	\label{sec:basic}
	
	We define some widely used definitions below. {For the set of integers $\ints$ and an integer $q$, define $\ints_q := \ints / q \ints$, i.e., $\ints_q$ is the set of integers $\pmod q$.} For $k \in \mathbb{N}$, we let {$[k]:=\{0,\ldots,k-1\}$}. All logarithms will be taken with respect to the base $2$. A function $\mu: \natural \rightarrow \mathbb{R}$ is said to be \emph{negligible} in a parameter $\lambda \geq 1$, which we denote $\negl(\lambda)$, if it satisfies the following:
	$$
	\text{for every integer } t>0, \text{ there exists an integer } K_t>0 \text{ such that for all } \lambda > K_t, \text{ we have } |\mu(\lambda)| < \lambda^{-t}.
	$$
	Similarly, a function $\eta: \natural \rightarrow \mathbb{R}$ is said to be
  \emph{non-negligible} in $\lambda$ if there exists an integer $t > 0$ and $K_t>0$ such
  that $\eta(\lambda) \geq \lambda^{-t}$ for all $\lambda>K_t$. For a distribution $D:\01^n\rightarrow
  [0,1]$, we write $x\sim D$ to say that $x$ is drawn according to the
  distribution $D$. For simplicity, we say that a function $f(n)=\poly(n)$, if there exist constants $a,b>0$ such that $n^a < f(n) < n^b$.%

  An algorithm $\A$ is said to have oracle access to a function
  $f$, which is denoted by $\A^f$, if $\A$ is allowed to (classically) query $f(x)$ for every
  $x$ in the domain of $f$ with unit computational cost. When $\A$ is a quantum
  algorithm, it is said to have quantum oracle access to $f$, denoted by
  $\A^{\ket{f}}$, if $\A$ is allowed to perform the operation $\sum_{x} \alpha_{x,b}
  \ket{x}\ket{b}\rightarrow \sum_{x,b} \alpha_{x,b}
  \ket{x}\ket{f(x)\oplus b}$ with unit computational cost for every $x\in \01^n,
  b\in \01$ and $\alpha_{x,b}\in \mathbb{C}$

  We say there is a \emph{non-negligible advantage} in
  distinguishing~$D$ from another distribution $D'$
  if there is a  $\poly(n)$-time adversary
  $\Adv$ (i.e., algorithm)  such that
    $$ 
  \Big|\Pr_{x_1,\ldots,x_L \sim D}[\Adv(x_1,\ldots,x_L) = 1] - \Pr_{x_1,\ldots,x_L \sim
  D'}[\Adv(x_1,...,x_L) = 1] \Big| \geq \eta(n),
$$
  where $\eta(n)$ is a non-negligible function and $L$ is a polynomial in $n$.

  Analogously, we say that there is a non-negligible advantage in
  distinguishing the functions $f, f' : \mathcal{X} \to \mathcal{Y}$
  if there is a $\poly(n)$-time adversary $\Adv$ (i.e., algorithm) that has
  oracle access to $f$ or~$f'$ such that
    $$ 
  \Big|\Pr[\Adv^{f}() = 1] - \Pr[\Adv^{f'}() = 1] \Big| \geq \eta(n),
$$
  where $\eta(n)$ is a non-negligible function. Here $n$ is some security
  parameters and $|\mathcal{X}|$, $|\mathcal{Y}|$ scale with $n$. If $\Adv$ instead has quantum oracle access to a function $f$, we
  denote it by $\Adv^{\ket{f}}$.
	
\vspace{-0.5em}
\subsection{Information theory and communication complexity}
We describe some basic concepts in information theory that we use later. Given a probability distribution $D:\mathcal{X}\rightarrow [0,1]$, the entropy of a random variable $X \sim D$ is given by
\begin{align*}
  \mbf{H}(X) = - \sum_{x \in \mathcal{X}} \Pr_{X \sim D}[X = x]\log \Big(\Pr_{X \sim
  D}[X=x]\Big).
\end{align*}

The \emph{binary entropy} of $\varepsilon\in [0,1]$ is defined as $\mbf{H}_b(\varepsilon)=-\varepsilon \log \varepsilon-(1-\varepsilon)\log(1-\varepsilon)$. Moreover, $\mbf{H}_b(\varepsilon)$ can be upper bounded as follows.
	\begin{fact}
		\label{fact:taylorseriesbinaryentropy}
		For all $\eps \in [0,1/2]$ we have the binary entropy $\mbf{H}_b(\eps)\leq O(\eps \log (1/\eps))$, and from the Taylor series expansion of $\mbf{H}_b(\varepsilon)$, we have
		$$
		1-\mbf{H}_b(1/2+\eps)\leq 2\eps^2/\ln 2+O(\eps^4).
		$$
	\end{fact}
\vspace{-0.8em}
	
Given a probability distribution $D:\mathcal{X} \times
\mathcal{Y}\rightarrow [0,1]$, and the random variables $(X,Y) \sim D$, the conditional entropy of
$X$ given $Y$ is 
\begin{align*}
  \mbf{H}(X|Y) = -\sum_{x \in \mathcal{X}, y \in \mathcal{Y}} \Pr_{(X,Y) \sim D}[(X,Y) = (x,y)]\log \Big(\frac{\Pr_{(X,Y) \sim D}[(X,Y) = (x,y)]}{\Pr_{X \sim D}[X=x]}\Big).
\end{align*}
Given a probability distribution $D:\mathcal{X} \times
\mathcal{Y} \times \mathcal{Z}\rightarrow [0,1]$, the random variables $(X,Y, Z) \sim D$, the 
conditional mutual information between $X$ and $Y$ given $Z$ is
\begin{align*}
  \mbf{I}(X: Y|Z) = \mbf{H}(X|Z) - \mbf{H}(X|Y,Z).
\end{align*}

The following fact about conditional entropy and prediction errors will be useful for us.

	\begin{lemma}[Fano's inequality]
		\label{lem:fano}
		Let $X$ be a random variable taking values in $\01$ and $Y$ be a random variables taking values in $\mathcal{Y}$. Let $f:\mathcal{Y}\rightarrow \01$ be a \emph{prediction} function, which predicts the value of $X$ based on an observation of $Y$. Suppose $\varepsilon=\Pr[f(Y)\neq X]$ is the probability of error made by the prediction function, then  $\mbf{H}(X\vert Y)\leq \mbf{H}_b(\eps)$.
	\end{lemma}
\vspace{-0.8em}

We now briefly describe communication complexity. For details, we refer the reader to \cite{Touchette15,KerenidisLGR16}. Here, we are interested in the setup with two parties Alice (denoted $A$) and Bob (denoted $B$). $A$ (resp.~$B$)
receives input $X$ (resp.~$Y$) such that $(X,Y) \sim D$ for a publicly known
probability distribution $D$. $A$ and $B$ then follow some protocol $\pi$ in which they exchange quantum information back-and-forth. Finally, $B$ outputs a random variable $Z$. The quantum communication complexity of the protocol $\mathsf{QCC}(\pi)$ is the number of qubits  communicated in the protocol~$\pi$.

Touchette~\cite{Touchette15} defined the notion of quantum information
complexity for a protocol, denoted $\mathsf{QIC}(\pi)$, which is rather subtle and out of the
scope of this work. In~\cite[Theorem~$1$]{Touchette15}, Touchette showed that
for all protocols $\pi$, we have $\mathsf{QIC}(\pi) \leq \mathsf{QCC}(\pi)$. Similarly, Kerenidis et al.~\cite[Theorem~$1$]{KerenidisLGR16} showed that $\mathsf{QIC}(\pi)$ is at most the classical information complexity of the protocol $\mathsf{CIC}(\pi)$,
whose definition we omit here. Also, it is not hard to see that if $B$ outputs
some value $Z$, then 
\begin{align*}
  \mbf{I}(Z:X|Y) \leq \mathsf{CIC}(\pi).
\end{align*}
Putting together \cite[Theorem~$1$]{Touchette15} and \cite[Theorem~$1$]{KerenidisLGR16} along with the inequality above, we obtain the following corollary.

\begin{corollary}
  \label{cor:communicationmutualinformation}
  Given a quantum communication protocol $\pi$ between two parties $A$ and $B$ whose inputs are $X$ and  $Y$, respectively, drawn from a distribution $D$. Let $Z$ be the output
  of $B$. Then, 
  $$\mbf{I}(Z:X|Y) \leq QCC(\pi).$$
\end{corollary}
\vspace{-1.2em}

\subsection{Cryptographic primitives}
\begin{definition}[Pseudo-random functions]
  \label{def:prf}
  Let $n$ be security parameter, $\mathcal{X}$, $\mathcal{K}$ and $\mathcal{Y}$ be finite sets
  and $\Fe = \{f_{\mathbf{k}} : \mathcal{X} \to \mathcal{Y}\}$ be family of functions
  indexed by $\mathbf{k} \in \mathcal{K}$.
  $\Fe$ is called a pseudo-random function family if, for
  every polynomial-time probabilistic
  algorithm $\Adv$ and for every constant $c>0$ we have 
	\begin{align*}
    \Big \vert \Pr_{f_{\mathbf{k}} \in \Fe} [\Adv^{f_\mathbf{k}}(\cdot)=1] - \Pr_{g \in \mathcal{U}}
    [\Adv^{g}(\cdot)=1]\Big| < \frac{1}{s^c},
	\end{align*}
  where $f_{\mathbf{k}}$ is picked uniformly at random from $\Fe$ and $g$ is picked
  uniformly at random from $\mathcal{U}$, the set of all functions
  from $\mathcal{X}$ to $\mathcal{Y}$. 
\end{definition}
\vspace{-0.8em}

In this work, we also consider  quantum polynomial-time
distinguishers. In short, when we say that a pseudo-random function is
\emph{secure}, we mean that it satisfies~\Cref{def:prf} and it is
\emph{quantum-secure} if it satisfies a variation of~\Cref{def:prf} where $\Adv$
is a polynomial-time quantum algorithm and has quantum query access to the functions
$f_{\mathbf{k}}$ and $g$. %
We point out that such a distinction is important.
 For example, there are examples of PRFs that are secure against (quantum)
 adversaries with classical oracle access to the function, but that are not
 secure when the adversary is allowed to perform quantum
 queries~\cite{BonehDFLSZ11}.
\vspace{-0.8em}

\subsection{Learning models}
	\label{sec:learningmodels}
	
	\subsubsection{Classical distribution-independent learning}
	\label{sec:classicaldistdependentlearn}
	We begin by introducing the classical Probably Approximately Correct (PAC)
  model of learning which was introduced by Leslie
  Valiant~\cite{valiant:paclearning}. A \emph{concept class} $\Cc$ is a
  collection of Boolean functions $c:\01^n\rightarrow \01$, which are often
  referred to as \emph{concepts}. In the PAC model, a learner~$\A$ is given access to a \emph{random example oracle} $\EX(c,D)$ where $c\in \Cc$ is an \emph{unknown} target concept (which the learner is trying to learn) and $D:\01^n\rightarrow [0,1]$ is an \emph{unknown} distribution. At each invocation of $\EX(c,D)$ the oracle returns a \emph{labelled example} $(x,c(x))$ where $x$ is drawn from the distribution~$D$.\footnote{Note that the oracle $\EX(c,D)$ doesn't take any input and simply returns a labelled example.} 
  Then $\A$ outputs a hypothesis $h$ and we say that $\A$ is an
\emph{$(\varepsilon,\delta)$-PAC learner} for a concept class~$\Cc$ if it satisfies the following: 
	\vspace{2pt}
	\begin{quote}
	for every $\varepsilon,\delta\in [0,1]$, for all $c\in \Cc$ and distributions
    $D$, when $\A$ is given $\varepsilon,\delta$ and access to the $\EX(c,D)$
    oracle, with probability $\geq 1-\delta$, $\A$ outputs a hypothesis $h$
    such~that 
    $$\Pr_{x\sim D}[h(x)\neq c(x)]\leq \varepsilon$$
	\end{quote}

The learner's \emph{advantage} over a random guess is given by  $\biaslearning = \frac{1}{2} - \frac{\eps}{2}$ and 
 $2\biaslearning = 1 - \eps$ is called the {\em bias} of the learner. %

	The \emph{sample complexity} of $\A$ is the maximum number of invocations of
  the $\EX(c,D)$ oracle which the learner makes  when maximized over all $c\in
  \Cc$ and all distributions $D$. Finally the $(\varepsilon,\delta)$-PAC \emph{sample
  complexity of $\Cc$} is defined as the minimum sample complexity over all $\A$
  that $(\varepsilon,\delta)$-PAC learn $\Cc$.  The \emph{time complexity} of
  $(\varepsilon,\delta)$-PAC learning $\Cc$ is the minimum number of \emph{time
  steps} of an algorithm $\A$ that $(\varepsilon,\delta)$-PAC learns $\Cc$
  (where the minimum is over all $\A$ that $(\varepsilon,\delta)$-PAC learn
  $\Cc$). We say that $\A$ $(\varepsilon,\delta)-$\emph{weakly learns} (resp.~strongly learns) $\Cc$ if $\varepsilon=\frac{1}{2}-n^{-c}$ (resp.~$\varepsilon=1/3$) for 
  some constant $c>0$ and input size $n$. Freund et al.~\cite{Freund99} showed that weak-PAC learning $\Cc$ is equivalent  to strong-PAC learning.
	
	\subsubsection{Quantum distribution-independent learning}
	\label{sec:quantumdistdependentlearn}
	The quantum model of PAC learning was introduced by Bshouty and Jackson~\cite{bshouty:quantumpac}. Instead of having access to an $\EX(c,D)$ oracle, here a \emph{quantum-PAC learner} has access to a $\QEX(c,D)$~oracle
	$$
	\QEX(c,D): \ket{0^n,0}\rightarrow \sum_x\sqrt{D(x)}\ket{x,c(x)},
	$$ 
	and we leave the $\QEX(c,D)$ oracle undefined on other basis states. We refer to the state produced by $\QEX(c,D)$ as a \emph{quantum example}, which is a coherent superposition over classical labeled examples. A quantum-PAC learner is given access to copies of quantum examples and performs a POVM (positive-valued-operator measurement), where each outcome of the POVM corresponds to a hypothesis. Similar to the classical distribution-independent learning setting, the \emph{quantum sample complexity} of an algorithm $\A$ is the maximum number of invocations of the $\QEX(c,D)$ oracle which the learner makes, when maximized over all $c\in \Cc$ and all distributions~$D$. The \emph{$(\varepsilon,\delta)$-quantum {\em PAC} sample complexity of $\Cc$} is defined as the minimum quantum sample complexity over all $\A$ that $(\varepsilon,\delta)$-quantum-PAC learn $\Cc$.

\vspace{-0.5em}	
	\subsubsection{Uniform distribution learning}
	\label{sec:learnunderuniform}
	The classical PAC model of learning places a strong requirement on learners, i.e., the learner needs to $(\varepsilon,\delta)$-PAC learn $\Cc$ for  \emph{every unknown distribution} $D$. In classical learning theory, there has been a lot of work in understanding a weaker model of learning -- when $D$ is restricted to the uniform distribution on $\01^n$ (which we denote as $\U$). In this restricted model, a classical learner is given access to  $\EX(c,\U)$ (known to the learner) which generates $(x,c(x))$ where $x$ is sampled according to the uniform distribution $\U$. An algorithm $\A$ is said to $(\varepsilon,\delta)$-learn $\Cc$ under $\U$ if it satisfies the following:
	\vspace{2pt}
	\begin{quote}
		for every $\varepsilon,\delta\in [0,1]$, for all $c\in \Cc$, when $\A$ is given $\varepsilon,\delta$ and access to the $\EX(c,\U)$ oracle, with probability $\geq 1-\delta$, $\A$ outputs a hypothesis $h$
		such~that 
		$\Pr_{x\sim \U}[h(x)\neq c(x)]\leq \varepsilon$ 
	\end{quote}

	The sample complexity and time complexity of learning $\Cc$ under the uniform distribution is defined similar to \Cref{sec:classicaldistdependentlearn} when we fix $D=\U$.

	One can similarly consider the case when a quantum learner is given access to $\QEX(c,\U)$
	$$
	\QEX(c,\U): \ket{0^n,0}\rightarrow \frac{1}{\sqrt{2^n}}\sum_x\ket{x,c(x)}.
	$$ 
	We leave $\QEX$ undefined on other basis states and assume the learner \emph{does not} have access to the inverse of $\QEX(c,\U)$.
	The quantum sample complexity and time complexity of learning a concept class $\Cc$ under the uniform distribution is defined similar to \Cref{sec:quantumdistdependentlearn} when we restrict $D=\U$. A powerful advantage of being given access to $\QEX(c,\U)$ is \emph{Fourier sampling}.  We do not introduce Fourier sampling here and refer the interested reader to~\cite[Section~2.2.2]{arunachalam:survey}.
	
	\subsubsection{Learning with membership queries}
	\label{sec:learnwithmembership}
	The classical model of PAC learning places yet another strong requirement on learners, i.e., the learner is only given access to labelled examples generated by the oracle $\EX(c,D)$ (for an unknown $c\in \Cc$ and distribution $D$). Angluin~\cite{angluin:exactmembership} relaxed this requirement and introduced the model of learning with \emph{membership queries}. In this scenario, in addition to $\EX(c,D)$, a learner is given access to a \emph{membership oracle} $\MQ(c)$ for the unknown target concept~$c$, which takes as input $x\in \01^n$ and returns $c(x)$. An algorithm $\A$ is said to $(\varepsilon,\delta)$-learn $\Cc$ under $D$ with membership queries if it satisfies the following:
	\begin{quote}
		for every $\varepsilon,\delta\in [0,1]$, for all $c\in \Cc$, when $\A$ is given $\varepsilon,\delta$ and access to $\EX(c,D)$, $\MQ(c)$ oracles, with probability  $\geq 1-\delta$, $\A$ outputs a hypothesis $h$ such that $\Pr_{x\sim D}[h(x)\neq c(x)]\leq~\varepsilon$.
	\end{quote}
	We abuse notation by saying the sample complexity of $\A$ is the maximum
  number of invocations of an oracle\footnote{Note that an invocation of either
  $\EX(c,D)$ or  $\MQ(c)$ counts as one application of the oracle.} when maximized over all $c\in \Cc$ under distribution $D$.
	The sample complexity and time complexity of learning $\Cc$ under $D$ given membership queries is defined similar to \Cref{sec:classicaldistdependentlearn} (where the classical learner now has access to $\MQ(c)$ in addition to $\EX(c,D)$).

	One can similarly consider the case when a quantum learner is additionally given access to a quantum membership query oracle $\QMQ(c)$
	$$
	\QMQ(c):\ket{x,b}\rightarrow \ket{x,b\oplus c(x)},
	$$
	for $x\in \01^n$ and $b\in \01$. The quantum learner is allowed to perform
  arbitrary unitary operations in between applications of the $\QMQ(c)$ oracle.
  The quantum sample complexity and time complexity of learning $\Cc$ under $D$
  given quantum membership queries is defined similar to
  \Cref{sec:quantumdistdependentlearn} (where the classical learner now has
  access to $\QMQ(c)$ in addition to $\QEX(c,D)$).
  In this paper, we will also
  view the concept as as $c\in \01^{2^n}$, described by the truth-table of
  $c:\01^n\rightarrow \01$.  In this case, we  say a learning algorithm is allowed to make \emph{queries to a string $c\in \01^N$}, it means  that the algorithm is given access to the oracle $\QMQ(c):\ket{x,b}\rightarrow \ket{b\oplus c_x}$ for $x\in [N]$  and $b\in \01$. 

	\paragraph{Learning with membership queries under the uniform distribution.}
  Finally, one can combine the learning models in this section with
  \Cref{sec:learnunderuniform} and consider (quantum) learners which are given
  access to (quantum) membership queries and (quantum) labelled examples when
  the underlying distribution $D$ is restricted to the uniform distribution
  $\U$. We say $\A$ is an \emph{$(\varepsilon,\delta)$-uniform PAC learner for
  $\Cc$} with membership queries if $\A$ $(\varepsilon,\delta)$-learns $\Cc$
  under $\U$ with membership queries. Similarly we can define a
  \emph{$(\varepsilon,\delta)$-uniform quantum-PAC learner for $\Cc$} with
  quantum membership queries. In this paper we will consider such learners in
  the weak and strong learning settings, and 
 for simplicity we will omit the $(\varepsilon,\delta)$-dependence when referring to the classical-PAC or quantum-PAC~learners. The sample complexity and time complexity of such learners is defined similar to \Cref{sec:classicaldistdependentlearn},~\ref{sec:quantumdistdependentlearn} (wherein the classical learner now has access to $\EX(c,\U)$ and $\MQ(c)$ and the quantum learner has access to $\QEX(c,\U)$ and $\QMQ(c)$). 
   	In order to further understand the theoretical aspects of quantum machine learning,  we refer the reader
	to~\cite{schuld:book,adcockea:qml,arunachalam:survey}.

\vspace{-0.5em}
\subsection{Circuits and neural networks}
	
	In this paper we will be concerned with the class of shallow or constant-depth Boolean circuits that consist of $\AND, \OR, \NOT$ and Majority gates. We define these classes formally now.
	
\vspace{-0.5em}	
	\subsubsection{The circuit classes \torp{$\ACz$}{AC0} and \torp{$\TCz$}{TC0}}
  An $\ACz$ circuit on $n$ bits consists of $\AND, \OR$ and $\NOT$ gates whose
  inputs are $x_1,\ldots,x_n,\overline{x_1},\ldots,\overline{x_n}$. Fan-in to
  the $\AND$ and $\OR$ gates are unbounded. The size of the circuit (i.e., the
  number of gates in the $\ACz$ circuit) is bounded by a polynomial in $n$ and
  the depth of the circuit is a constant (i.e., independent of $n$). Furthermore, the gates at level $i$ of the circuit have all their inputs coming from the $(i-1)$-th level and all the gates at the same level are $\AND$ or $\OR$ gates. The class of Boolean functions that can be expressed by such a depth-$d$ circuit is written as $\ACz_d$ and $\ACz=\bigcup_{d\geq 0} \ACz_d$. 
	A depth-$d$ threshold circuit, denoted $\TCz_d$, is similar to an $\ACz_d$
  circuit, except that the circuit is also allowed to have Majority gates
  $\MAJ:\01^n\rightarrow \01$, where $\MAJ(x)=1$ if and only if $\sum_i
  x_i\geq n/2$. 
  
Additionally, we will need the following definition of halfspaces and polynomial threshold functions in order to discuss about depth-$2$ $\TCz_2$ circuits. A \emph{half-space} in $n$ dimensions is a Boolean function $f:\01^n\rightarrow \01$ of the form 
$$
f(x)= \indic{\sum_i a_ix_i\geq \theta}
$$
where $a_1,\ldots,a_n$ and $\theta$ are fixed integers and $\indic{\cdot}$
denotes the indicator function which evaluates to $1$ if and only if $\sum_i
a_ix_i\geq \theta$.  The \emph{intersection of $k$ half-spaces} is a function
of the form $g(x)=\bigwedge_{i=1}^k f_i(x)$ where the $f_i$s are half-spaces
and $\wedge$ is the $\mathsf{AND}$ function. A \emph{polynomial threshold
	function} (PTF) of degree $d$ is a Boolean function of the form
$f(x)=\indic{p(x)\geq 0}$, where $p$ is a degree-$d$ polynomial with
integer coefficients. Note that a half-space is a degree-$1$ PTF. A PTF~$f$,
defined as $f(x)=\indic{p(x)\geq 0}$, is called \emph{light} if the sum of
the absolute value of the coefficients in $p$  is at most polynomial in $n$. Using this, we now have the following claim.

\begin{claim}
	\label{claim:halfspaceintc0}
A polynomially-light halfspace on $n$ dimensions can be computed in $\TCz_2$.
\end{claim}

\begin{proof}
 First observe that each half-space on $n$ bits is already a majority gate (with the inputs suitably negative and replicated based on the coefficients $a_1,\ldots,a_n$) and the top gate $\mathsf{AND}(f_1,\ldots,f_t)$ can be replaced by the majority gate $\mathsf{MAJ}(-t,f_1,\ldots,f_t)$, which is a depth-$2$ threshold circuit.
\end{proof}

\vspace{-2em}
\paragraph{Neural networks and $\TCz$.} 
  One motivation for learning the
  concept class $\TCz$ is that this class is a theoretical way to model neural networks. Although we do not deal with neural networks  in this paper, we briefly mention their connection to $\TCz$ circuits. A \emph{feed-forward neural network} can be modeled as an acyclic directed graph where the nodes consist of \emph{neurons}.\footnote{In general, neural networks are \emph{cyclic} directed graphs. In order to make the connection to $\TCz$ we consider a subclass of neural networks called feed-forward neural-networks.} Every neuron is associated with some internal \emph{weights} $w_0,\ldots,w_n \in \R$. The action of the neuron is defined as follows: it takes as real input $x_1,\ldots,x_n$ and generates an output signal $y$ defined by 
	$$
	y=\sigma \Big(\sum_{i\in [n]}w_i x_i -w_0\Big),
	$$
	where $\sigma$ is a \emph{transfer function}. The size of a neural network is the number of neurons in the network and the depth of the network is the number of \emph{layers} in the network between the input and final output. A commonly used transfer function is the sigmoid function, defined as $\sigma(t)=(1+e^{-t})^{-1}$. The sigmoid function is a continuous approximation of the threshold function $\mathsf{Thr_{w_0,w_1,\ldots,w_n}}$, which on input $x_1,\ldots,x_n$, outputs $1$ if $\sum_iw_ix_i\geq w_0$ and $0$ otherwise. Maass, Schnitger and Sontag~\cite{maas:sigmoidboolean} showed an equivalence between feed-forward neural networks and $\TCz$ circuits. Before we state their theorem, we first define the following:
	
\begin{definition} 
Let $\varepsilon>0$ and $f:\01^n\rightarrow \01$ be a Boolean function. We say a feed-forward neural network \emph{$C$ $\varepsilon$-computes $f$}, if there exists a $t_C\in \R$ such that: on input $x\in \01^n$, $C(x)\geq t_C+\varepsilon$ if~$f(x)=1$ and $C(x)\leq t_C-\varepsilon$ if $f(x)=0$.
\end{definition}

We now state the main result of~\cite{maas:sigmoidboolean}.

\begin{theorem}[Theorem~4.1\cite{maas:sigmoidboolean}]
Let $d\geq 1$. Let $\mathsf{T}_d$ be the class of functions $f:\01^n\rightarrow \01$ such that there exists a depth-$d$ feed-forward neural network $C$ on $n$ bits, whose size is $\poly(n)$, the individual weights in $C$ have absolute value at most $\poly(n)$ and $C$ $\frac{1}{3}$-computes $f$. Then, $\mathsf{T}_d=\TCz_d$.\footnote{The $1/3$ can be replaced by an arbitrary constant independent of $n$.} 
\end{theorem}
	
	In particular, this theorem implies that Boolean functions that are computable by constant-depth polynomial-sized neural networks with neuron-weights bounded by some polynomial in the input size, can be implemented as a circuit in $\TCz$. So an alternate definition of $\TCz$ is the class of constant-depth neural networks with polynomial size and weights bounded polynomially in the input length.

\section{Hardness assumptions and cryptographic constructions}
\label{sec:hardnessofPRF}
In this section we give a detailed introduction to lattice-based problems and some cryptographic primitives built using them. Our motivation in doing this, firstly, is to highlight the subtleties inherent in the hardness results and security proofs associated with these lattice-based problems and primitives respectively.  Secondly, the details also help in explaining the subtleties that are reflected in our utilization of these primitives for showing the hardness of quantum learning. 
\subsection{Learning with Errors (LWE)}

Throughout this section, we will need the following notation. Let $\chi$ be a distribution over $\ints$ and assume that we can efficiently sample from $\chi$.\footnote{The distribution $\chi$ is dependent on input size $n$, but we drop the $n$-dependence for notational simplicity.} Let $B>0$. We say that $\chi$ is $B$-bounded if $\pr_{e \sim \chi}[ |e| > B] \leq {\mu(n)}$, where $e, B$ are at most $n$ bits long and {$\mu$ is a function that is negligible in $n$}. In other words, a distribution $\chi$ is $B$-bounded if, with high probability, the magnitude of $e$ when drawn according to $\chi$ is at most~$B$. Also, a \emph{discrete Gaussian} refers to the standard normal Gaussian distribution on the real line which is restricted to take integer values. Now, we can define the decision version of the \emph{Learning with Errors ($\LWE$)} problem. %

\begin{definition}[$\LWE_{d,q,\chi, m}$] %
	\label{def:LWE}
	The (decision) Learning with Errors problem with dimension $d$, modulus $q$
  and a $B$-bounded distribution $\chi$ over $\ints$ is defined as follows: On
  input {$m$ independent samples} $\{(\mathbf{a}_i, b_i) \in \ints_q^d \times
  \ints_q\}_i$, where the $\mathbf{a}_i$ are uniformly sampled from $\ints_q^d$, distinguish {(with non-negligible advantage)} between the following two cases: %
	\begin{itemize}
	\setlength{\leftmargin}{+3in}
		\item ($\LWE$-samples) The $b_i$s are \emph{noisy products} with respect to a fixed secret $\mathbf{s}$ distributed uniformly in $\ints_q^d$, i.e., $b_i$s are of the form $b_i = \mathbf{a}_i \cdot \mathbf{s} + e_i \Mod q$ where $e_i\in \ints$ is sampled according to  $\chi$ conditioned on $|e_i| \leq B$.
		\item (Uniform samples) For every $i$, $b_i$ is uniformly sampled from $\ints_q$ and is independent of $\mathbf{a}_i$. %
	\end{itemize}
\end{definition}
 \vspace{-0.8em}
When the number of samples $m$ is arbitrary (i.e., not bounded in terms of the dimension $d$), we will simply denote the problem as $\LWE_{d, q, \chi}$.

Consider an efficient distinguisher $\De$ as defined in \Cref{def:prf} that
distinguishes a distribution $D$ from being uniformly random. Then, if $D$ is
the distribution generated by $\LWE$-samples, the distinguisher $\De$ would
clearly serve as an efficient algorithm to solve the (decision) $\LWE$ problem. We formally define this below.

\begin{definition}[Distinguishers for $\LWE_{d,q,\chi, m}$]
\label{defn:quantumdistinguisher}
	An algorithm $\De$ is a distinguisher for the (decision) $\LWE_{d, q, \chi, m}$ problem if, given $m$ independent samples from $\ints_q^d \times \ints_q$, it distinguishes $m$ $\LWE$-samples ${\{\mathbf{a}_i,b^{(1)}_i\}_i}$ from $m$ uniformly random samples ${\{\mathbf{a}_i,b^{(2)}_i\}_i}$ with non-negligible advantage i.e., there exists a non-negligible function $\eta:\natural \rightarrow \mathbb{R}$ such that
	$$
	\Big \vert \Pr_{(\mathbf{a}_i,b^{(1)}_i)} [\De(\{\mathbf{a}_i,b^{(1)}_i\}_i)
  \text{ outputs 1}] -\Pr_{(\mathbf{a}_i,b^{(2)}_i)}
  [\De(\{\mathbf{a}_i,b^{(2)}_i\}_i) \text{ outputs 1}] \Big \vert \geq \eta(d). %
	$$
	When $\De$ is a $\poly(d, m)$-time probabilistic (resp.~quantum) algorithm, it is said to be an \emph{efficient} classical (resp.~quantum) distinguisher. 
	\end{definition}

\vspace{-0.8em}
\subsubsection{Hardness of LWE} 
\label{sec:hardnesslwe}
 For a suitable choice of parameters, it is believed that $\LWE$ is a
 hard problem to solve. This is based on the worst-case hardness of lattice-based problems such as
 $\mathsf{GapSVP}_\gamma$ (decision version of the shortest vector problem) or
 $\mathsf{SIVP}$ (shortest independent vectors
 problem)~\cite{regev:lattice,peikert:publicSVP}. We do not introduce these
 problems or discuss their hardness here; an interested reader can refer to Peikert's survey~\cite{Peikert16} on the topic. 

The first hardness result for $\LWE$ was a \emph{quantum reduction} to the
$\mathsf{GapSVP}$ problem~\cite{regev:lattice,peikert:publicSVP}. Subsequent
works succeeded in de-quantizing the reduction from~\cite{regev:lattice} except
that the modulus $q$ was required to be super-polynomial in the dimension
$d$~\cite{peikert:publicSVP}. Later, Brakerski et al.~\cite{BLPRS13} built on
this result and other techniques from fully homomorphic encryption,  obtaining a reduction with a polynomial~modulus.
\begin{theorem}[{Hardness of LWE} \cite{BLPRS13}]  
\label{thm:classred}
Let $d,q\geq 1$ and $\alpha \in (0,1)$ be such that $q = \poly(d)$ and $\alpha q\geq 2 d$. Let $D_{\ints_q,\alpha}$ be the discrete Gaussian distribution over $\ints_q$ with standard deviation $\alpha q$. Then there exists a classical reduction from the worst-case hardness of the $d$-dimensional $\mathsf{GapSVP}_{\widetilde{O}(d/\alpha)}$ problem to the $\LWE_{d^2, q, D_{\ints_q,\alpha}}$ problem.
\end{theorem}
\vspace{-0.8em}
The corollary below provides a suitable choice for parameters that lead to hard $\LWE$ instances. 

\begin{corollary}
	\label{cor:redlwetogapsvp}
	Let $d \in \natural$ and $\alpha, q$ be parameters such that {$\alpha =
  1/\sqrt{d}$} and $\alpha q \geq 2\sqrt{d}$. Let $\chi=D_{\ints_q,\alpha}$,
  the discrete Gaussian distribution over $\ints_q$ with standard deviation
  $\alpha q$. If there exists a polynomial-time quantum distinguisher for
  $\LWE_{d, q, \chi, m}$, then there exists a polynomial-time quantum algorithm for the $\sqrt{d}$-dimensional $\mathsf{GapSVP}_{\widetilde{O}({d^{1.5}})}$ problem.
\end{corollary}
\vspace{-0.8em}

 It is believed that there exists no polynomial-time
  quantum algorithms for the $d$-dimensional
  $\mathsf{GapSVP}_{\widetilde{O}(d^3)}$ problem~\cite{GG00}, which implies that 
  there are no
  quantum polynomial-time algorithms for $\LWE$ with the parameters as stated in
  \Cref{cor:redlwetogapsvp}.

 We would like to stress that we show hardness results for learning
 constant-depth circuits with quantum examples, based on the hardness of solving
 $\LWE$  with {\em classical samples}. This does not
	contradict the work by Grilo et al.~\cite{grilo:LWEeasy} who prove that $\LWE$ is easy when {\em quantum examples} are provided.
  
  \vspace{-0.8em}
  
\subsubsection{Public-key encryption scheme based on LWE}
\label{sec:cryptosystemLWE}
  In this section we describe the public-key cryptosystem proposed by
  Regev~\cite{regev:lattice} whose security is based on the hardness of $\LWE$. 

  \begin{definition}[\LWEPKE{} \cite{regev:lattice}] The $\LWEPKE_{d,q,m}$ public-key
    encryption scheme consists~of:
      
    \noindent \textbf{Key-generation:} Pick $\textbf{s}\in \ints_q^d$ and $A\in \ints_q^{m \times
    d}$ uniformly at random from the respective supports. Draw $e \in \ints_q^d$ from the distribution $\chi$, as defined in the 
    \LWE{} problem. Let $\textbf{b} = A\textbf{s} + e$ and output $\Kpriv = \textbf{s}$
    and $\Kpub = (A, b)$.

    \noindent \textbf{Encryption:} To encrypt a bit $c$ using $\Kpub = (A, \textbf{b})$, pick $S
    \subseteq [m]$ uniformly at random and the encryption is $\big(\mathbf{1}_S^T \cdot A, \mathbf{1}_S^T \cdot  \textbf{b} + c \lfloor
    \frac{q}{2} \rfloor\big)$, where $\mathbf{1}_S\in \01^{m}$ is defined as $\mathbf{1}_S(i)=1$ if and only if $i\in S$

    \noindent \textbf{Decryption:} In order to decrypt the ciphertext $(a,b)$ using
    $\Kpriv = \textbf{s}$, output $0$
    if $b - a^T \textbf{s} \Mod{q} \leq~\lfloor\frac{q}{4}\rfloor$, otherwise output $1$.
  \end{definition}

  \begin{theorem}
\label{thm:regevsvphardnesscrypto}
 {
    Let $d \in \natural$, $\eps>0$ be a constant, $q=\poly(d)$
    and $m\geq (1+\eps)(d+1)\log{q}$ be polynomially bounded in $d$. 
    Then, the probability of decryption error for $\LWEPKE_{d,q,m}$ is
    $2^{-\omega(d^2/m)}$. Moreover, an adversary can
    distinguish an encryption of $0$ from an encryption of $1$ in
    polynomial time with
    non-negligible advantage over a random guess iff there is a polynomial-time distinguisher for
    $\LWE_{d,q,\chi,m}$.}
  \end{theorem}
  \vspace{-0.8em}

  The following property of $\LWEPKE{}$ was proven by Klivans and Sherstov~\cite{KlivansS09}.
	
	\begin{lemma}[Lemma~4.3~\cite{KlivansS09}]
		\label{lem:SVPindeg2ptf}
    The decryption function of \LWEPKE{} can be computed by light degree-$2$ PTFs.
	\end{lemma}
\vspace{-0.8em}

\subsection{Ring-LWE}
\label{sec:rlwe}

One of the main drawbacks of $\LWE$-based cryptographic primitives is that their representations are not compact, which makes them inefficient to implement. To overcome this, Lyubashevsky et al.~\cite{LPR10} proposed a variant of $\LWE$ over polynomial rings, called \emph{Ring-$\LWE$}. 
{Unfortunately, the efficiency of implementing Ring-$\LWE$ does not
come for free: first, the average-case hardness of Ring-$\LWE$ is currently supported only
by the worst-case hardness of problems on ideal lattices; and secondly
the reduction is still
(partly) \emph{quantum} (see
Section~\ref{sec:hardnessofrlwe} for further discussion).} %

The following notation will be used whenever we refer to the ring variants of
lattice problems. Let $R$ be a degree-$d$ \emph{cyclotomic ring} of
the form ${R:=\ints[X]/\langle X^d+1\rangle}$ where $d$ is a power of $2$. In particular, the elements of $R$ will be represented by their residues modulo $(X^d + 1)$ -- which are integer polynomials of degree less than $d$. For an integer
modulus $q$, we let $R_q := {R / qR = \ints_q[X] / \langle X^d + 1 \rangle}$ where its
elements are canonically represented by integer polynomials of degree less than
$d$ with coefficients from $\ints_q$.

{Let $\Upsilon$ denote an efficiently sampleable distribution over the elements of $R$ that is concentrated on elements having $B$-bounded integer coefficients, i.e., let $\Upsilon$ be a variant of the $d$-dimensional discrete Gaussian distribution where each coefficient of the integer polynomial (in some canonical representation) is $B$-bounded with overwhelmingly high probability. {As with $\LWE$, the Ring-$\LWE$ problem focuses on finding a secret ring element $\mbf{s} \in R_q$ when given noisy \emph{ring} products $(\mbf{a}, \ \mbf{a} \cdot \mbf{s} + e)$ as samples where $\cdot$ and $+$ refer to the multiplication and addition operations in $R_q$.}

\begin{definition}[$\RLWE_{d,q,\Upsilon, m}$] %
	\label{def:RLWE}
  The (decision) Ring Learning with Errors problem with dimension $d$, modulus $q$ and a $B$-bounded distribution $\Upsilon$ over $R$ is defined as follows: On
  input {$m$ independent samples} $\{(\mathbf{a}_i, b_i) \in R_q \times
  R_q\}_i$, where the $\mathbf{a}_i$ are uniformly sampled from $R_q$, distinguish {(with non-negligible advantage)} between the following two cases: %
	\begin{itemize}
	\setlength{\leftmargin}{+3in}
		\item ($\RLWE$-samples) The $b_i$s are \emph{noisy products} with respect to
      a fixed secret $\mathbf{s}$ distributed uniformly in $R_q$, i.e., $b_i$s
      are of the form $b_i = \mathbf{a}_i \cdot \mathbf{s} + e_i \Mod q$ where
      $e_i\in R$ is sampled according to  $\Upsilon$.
		\item (Uniform samples) For every $i$, $b_i$ is uniformly sampled from $R_q$ and is independent of $\mathbf{a}_i$. %
	\end{itemize}
\end{definition}

Note that, similar to~\Cref{defn:quantumdistinguisher}, we can also consider
distinguishers, $\De$, for the ring-LWE problem $\RLWE_{d,q,\Upsilon,m}$. At times, we will consider the \emph{normal form} of the problem where the secret is sampled from the error distribution modulo $q$ i.e., $\mbf{s} \leftarrow \chi^d \Mod q$. It is known, from~\cite{LPR13}, that this form of the problem is as hard as the  problem in~\Cref{def:RLWE}. %

\vspace{-0.5em}
\subsubsection{Hardness of Ring-LWE}
\label{sec:hardnessofrlwe}

Lyubashevsky, et. al.~\cite{LPR10} first showed a worst-case hardness reduction
from the shortest vector problem $\mathsf{SVP}_{\gamma}$ on a
class of structured lattices called \emph{ideal lattices} to Ring-$\LWE$. It is believed
that $\mathsf{SVP}_{\gamma}$ is hard in the worst case, even for \emph{quantum}
algorithms. For instance, the best known quantum algorithms for
{(general)} $\mathsf{SVP}_{\poly(d)}$ still run in time $2^{O(d)}$. Although the hardness reduction between Ring-$\LWE$ and $\mathsf{SVP}_\gamma$ uses ideal lattices, the hardness assumptions for $\mathsf{SVP}_{\poly(d)}$ are still fairly standard. One caveat here is that that reduction poses restrictions on the types of rings and requires that the modulus $q = \poly(d)$. For more on this subject, we direct the interested reader to Peikert's survey~\cite{Peikert16}.

Subsequently, Langlois and Stehl\'{e}~\cite{LS15} improved the hardness
result for any modulus $q$ with $\poly(d)$ bit-size (therefore $q$ can be
super-polynomial in $d$). More recently, Peikert
et al.~\cite{PRS17} further improved the hardness result to work for every ring
and modulus by obtaining a direct reduction from the worst-case hardness of the
$\mathsf{SIVP}_{\gamma}$ (find approximately short independent vectors in a
lattice of length within a $\gamma$-factor of the optimum) {on ideal
lattices} to the (decision) $\RLWE_{d,q,\Upsilon,m}$ problem. We use the following result from their work applied to cyclotomic rings.

\begin{theorem}[Quantum Reduction from \cite{PRS17}]
\label{cor:rlwe-prs}
Let $R$ be a degree-$d$ cyclotomic ring and let $m = poly(d)$. Let $\alpha =
  \alpha(d) \in (0, 1)$ and $q = q(d) \geq 2$ be an integer such that $\alpha
  \leq 1/2 \sqrt{\log d / d}$ and $\alpha q \geq \omega(1)$. Let $\Upsilon$ be a
  $d$-dimensional discrete Gaussian distribution over $R$ with standard
  deviation at most $\alpha q$ in each coefficient. Then, there exists a
  polynomial-time quantum reduction from the $d$-dimensional $\mathsf{SIVP}_{\gamma}$ over $R$ to the
  $\RLWE_{d,q,\Upsilon,m}$ problem for any $\gamma \leq \max\{\omega(\sqrt{d\log d}/\alpha), \sqrt{2}d\}$. 

\end{theorem} 
\vspace{-0.5em}

When $\gamma = \widetilde{O}(\poly(d))$, the best known quantum algorithms for $\mathsf{SIVP}_{\gamma}$ are still expected to take exponential time and any sub-exponential time algorithm for these approximation factors will be considered a breakthrough in the cryptography and quantum community.
Based on this, we will use the following three assumptions on the hardness of $\RLWE$ in the remainder of this work.

\begin{assumption}[Polynomial time hardness]\label{assumption:polynomial-hardness}
  Let $0 < \tau < 1$, $d \in \mathbb{N}$, $q = 2^{O(d^\tau)}$, $r \leq \poly(d)$ and $\Upsilon(r)$ be the discrete Gaussian distribution over degree-$d$ cyclotomic rings with standard deviation at most $r$ in each coefficient.
		The $\RLWE_{d,q,\Upsilon(r)}$ problem is hard for a $\poly(d)$-time quantum
    computer.
\end{assumption}

\vspace{-0.7em}
\noindent \Cref{assumption:polynomial-hardness} is used in \Cref{cor:tc0-hard} to show
that $\TCz$ cannot be learned in polynomial time.

\begin{assumption}[Quasi-polynomial time hardness]\label{assumption:quasipoly-hardness}
  Let $0 < \tau < 1$, $d \in \mathbb{N}$, $q = 2^{O(d^\tau)}$, $r \leq \poly(d)$ and $\Upsilon(r)$ be the discrete Gaussian distribution over degree-$d$ cyclotomic rings with standard deviation at most $r$ in each coefficient.
		The $\RLWE_{d,q,\Upsilon(r)}$ problem is hard for a $2^{O(\polylog(d))}$-time quantum
    computer.
\end{assumption}

\vspace{-0.7em}
\noindent \Cref{assumption:quasipoly-hardness} helps to show
that $\TCz$ cannot be learned in quasi-polynomial time in~\Cref{cor:tc0-hard}.

\begin{assumption}[Strongly sub-exponential time hardness]\label{assumption:subexp-hardness}
  Let $0 < \tau < 1$, $0  < \eps < \frac{1}{2}$ , $d \in \mathbb{N}$, $q = 2^{O(d^\tau)}$, $r \leq d^\lambda$ for $0 < \lambda < 1/2 - \eps$ and $\Upsilon(r)$ be the discrete Gaussian distribution over degree-$d$ cyclotomic rings with standard deviation at most $r$ in each coefficient.
		For every constant $\eta > 2$, the $\RLWE_{d,q,\Upsilon(r)}$ problem is hard for a $2^{o(d^{1/\eta})}$-time
    quantum computer.
\end{assumption}

\vspace{-0.7em}
\noindent \Cref{assumption:subexp-hardness} is used in \Cref{cor:ac0-hard} to
show that $\ACz$ cannot be learned in polynomial time.

\subsection{Learning with Rounding (LWR)}
\label{sec:lwrintro}
While the hardness results for $\LWE$ and $\RLWE$ suggest that they are good candidates for
quantum-secure cryptographic primitives, the amount of randomness these schemes need to
generate samples make them unsuitable for the construction of pseudo-random
{functions}. For this reason, we consider the Learning with Rounding problem ($\LWR$) which
was introduced by Banerjee, Peikert and Rosen~\cite{BPR12} and can be seen as
a deterministic variant of $\LWE$. 

The following notation will be used in this section. For $\mathbf{a},\mathbf{s}\in \ints_q^d$, let $\mathbf{a} \cdot \mathbf{s}$ be defined as $\sum_{j\in [d]} a_{j}s_j \mod q$. For $q \geq p\geq 2$, we define $\lfloor \cdot\rceil_p:\ints_q\rightarrow \ints_p$ as $\lfloor x\rceil_p= \lfloor (p/q)\cdot x\rceil$ where $\lfloor \cdot \rceil$ denotes the closest integer. Additionally, when $\mbf{x}$ is a ring element in $R_q$, we extend this rounding operation to rings as follows: define $\lfloor \cdot \rceil_p:R_q \rightarrow R_p$ as the operation where each polynomial coefficient $x_i \in \ints_q$ of $\mbf{x}$ is rounded to $\ints_p$ as $\lfloor x_i \rceil_p$ respectively. From here onwards, we assume 
that $q$ has a $\poly(d)$ bit size, i.e. $\log q \leq O(\poly(d))$, and, 
without loss of generality, that $p$, and $q$ are powers of~$2$.  

\begin{definition}[$\LWR_{d,q,p,m}$]%
\label{def:lwr}
	The (decision) Learning with Rounding problem with dimension $d$, modulus $q$ and rounding modulus $p<q$ is defined as follows: on input {$m$ independent samples} $\{(\mathbf{a}_i, b_i) \in \ints_q^d \times \ints_p\}_i$ where the $\mathbf{a}_i$'s are sampled uniformly from $\ints_q^d$, distinguish {(with non-negligible advantage)} between the following two cases: %
	\begin{itemize}
		\item ($\LWR$-samples) There is a fixed secret $\mathbf{s}$ uniformly sampled from $\ints_q^d$ such that the $b_i$s are \emph{rounded products} with respect to $\mathbf{s}$, i.e., $b_i = \lfloor \mathbf{a}_i \cdot \mathbf{s}\rceil_p$ for every $i$.

		\item (Uniform samples)  For every $i$, $b_i$ is uniformly sampled from $\ints_p$ and is independent of $\mathbf{a}_i$. 
	\end{itemize}
\end{definition} 
\vspace{-0.8em}

As $\LWR$ inherits all issues of inefficient implementation exhibited by $\LWE$, we are particularly interested in the ring version of this problem namely, Learning with Rounding over Rings. This is defined below.

\begin{definition}[$\RLWR_{d,q,p,m}$]%
\label{def:rlwr}
	The (decision) Learning with Rounding problem with dimension $d$, modulus $q$ and rounding modulus $p<q$ is defined as follows: on input {$m$ independent samples} $\{(\mathbf{a}_i, \mbf{b}_i) \in R_q \times R_p\}_i$ where the $\mathbf{a}_i$'s are sampled uniformly from $R_q$, distinguish {(with non-negligible advantage)} between the following two cases: %
	\begin{itemize}
		\item ($\RLWR$-samples) There is a fixed secret $\mathbf{s}$ uniformly sampled from $R_q$ such that the $b_i$s are \emph{rounded ring  products} with respect to $\mathbf{s}$, i.e., $\mbf{b}_i = \lfloor \mathbf{a}_i \cdot \mathbf{s} \rceil_p$ for every $i$.

		\item (Uniform samples)  For every $i$, $b_i$ is uniformly sampled from $R_p$ and is independent of $\mathbf{a}_i$. 
	\end{itemize}
\end{definition} 
\vspace{-0.8em}

While the definition above samples a secret uniformly from $R_q$, an equivalent problem samples $\mbf{s}$ using the $\RLWE$ error distribution $\Upsilon \Mod{q}$. This will have a bearing when we consider pseudo-random functions constructed from $\RLWR$. When the number of samples is arbitrary, the corresponding problems are denoted as $\LWR_{d, q, p}$ and $\RLWR_{d,q,p}$. Similar to \Cref{defn:quantumdistinguisher} of distinguishers for $\LWE_{d,q,\chi, m}$ one could analogously define classical and quantum distinguishers for the $\LWR_{d,q,p,m}$ and $\RLWR_{d,q,p,m}$ problems.

\subsubsection{PRFs from the LWR problem}
\label{sec:lwr-prg}

The hardness reduction in~\cite{BPR12} from $\RLWR_{d,q,p,m}$ (resp.
$\LWR_{d,q,p,m}$) to $\RLWE_{d,q,\Upsilon}$ (resp. $\LWE_{d,q,\chi}$) for $q =
\superpoly(d)$ implies that these problems may be hard for quantum computers.
This makes them good candidates from which to construct quantum-secure PRFs. 
The following PRF construction from~\cite[Section~5]{BPR12} will be of importance to us. %

\begin{definition}[$\RLWE$ degree-$k$ PRF]
\label{def:LWEPRF} 
  For parameters $d \in \natural$, moduli $q \geq p \geq 2$ and input length $k \geq 1$, let $R$ be a degree-d cyclotomic ring. The $\RLWEPRF_{d,q,p,k}$ function family $\RF$ is defined as $\RF := \{ f_{\mbf{a}, \mbf{s}_1, \ldots, \mbf{s}_k} : \01^k \rightarrow R_p \}$ where $\mbf{a} \in R_q$, $\mbf{s}_i \in R_q$ for $i \in [k]$ and 
  \[f_{\mbf{a},\mbf{s}_1,...,\mbf{s}_k}(x) := \left\lfloor \mbf{a} \cdot \Pi_{i \in
  [k]} \mbf{s}_i^{x_i}\right\rceil_p\]
  where $\left\lfloor \cdot \right\rceil_p$ denotes the rounding procedure over ring elements.
\end{definition}

Clearly, every function in the family $\RF$ defined above, outputs a ring element from $R_p$.
However, for the purposes of constructing the
simplest concept class that is hard to learn, we would prefer a PRF that outputs a
single bit. We can achieve this by modifying the construction to output the most significant bit of the PRFs output from~\Cref{def:LWEPRF} and denote it as the One-Bit RLWR PRF ($\oBRPRF$) as defined below. 

\begin{definition}[One-Bit $\RLWR$ PRF]
\label{def:RLWRPRF} 
  For parameters $d \in \natural$ and moduli $q \geq p \geq 2$, let $R$ be a degree-$d$ cycolotomic ring. Let $\phi_{p} : R_q \rightarrow \01$ be defined as $\phi_{p} (\mbf{g}) = \msb\left(\left\lfloor \mbf{g} \right\rceil_p \right)$ where $\msb(\cdot)$ for a ring element denotes the most significant bit of it's first polynomial coefficient. For an input length $k \geq 1$, the $\oBRPRF_{d,q,p,k}$ function family $\RF$ is defined as $\RF := \{ f_{\mbf{a}, \mbf{s}_1, \ldots, \mbf{s}_k} : \01^k \rightarrow \01 \}$ where $\mbf{a}, \mbf{s}_1, \ldots, \mbf{s}_k\in~R_q$ and 
  \[f_{\mbf{a},\mbf{s}_1,...,\mbf{s}_k}(x) = \phi_{p} \left(\mathbf{a} \cdot \Pi_{i \in [k]} \mathbf{s_i}^{x_i} \right).\]
\end{definition}

Note that the definition above is equivalent for $\phi$ outputting any arbitrary fixed bit from the representation of the ring element. 
~\cite{BPR12} showed that the $\RLWE$ degree-$k$ PRFs are both secure against classical adversaries. 
Combining techniques from~\cite{BPR12} and~\cite{zhandry:qprf},\footnote{Zhandry~\cite{zhandry:qprf} proved that the $\LWE$ degree-$k$ PRF is also secure against quantum adversaries. This was done by showing that the classical security of the PRF under the $\LWE$ assumption implies the quantum security under the same assumption.} one can conclude in a fairly straightforward manner that the $\RLWE$ degree-$k$ PRF is also secure against quantum adversaries under the $\RLWE$ assumption. For completeness, we show how this security proof also implies the security of the One-Bit $\RLWR$ PRFs.

\begin{lemma}
\label{lem:prf_sec}
  If the $\RLWE$ degree-$k$ \emph{PRF} is quantum-secure, then so is the
  One-Bit $\RLWR$ \emph{PRF}.
\end{lemma}

\vspace{-0.5em}
\begin{proof}
  Let us assume that there exists a distinguisher $\De$ for the One-Bit
  $\RLWR$ PRF.
  This means that $\De$ can make (quantum) queries to an oracle $O$ and can
  distinguish if $O$ is a function picked from the One-Bit $\RLWR$ PRF family or a uniformly
  random Boolean function. %

  We construct a distinguisher $\De'$ for the $\RLWE$ degree-$k$
  PRF that uses $\De$ as a sub-routine. By definition, $\De'$ has (quantum) query access to an oracle $O'$
  that is either a function picked from the $\RLWE$ degree-$k$ PRF family or 
  a uniformly random function that outputs an element from $R_p$. %

In order to use $\De$ as a sub-routine successfully, $\De'$, using queries to its oracle $O'$, 
  simulates an oracle $\widetilde{O}$ that will answer the (quantum) queries from $\De$. 
  In this case, $\widetilde{O}$ is a function from One-Bit $\RLWR$ PRF when $O'$ 
  is a function from $\RLWE$ degree-$k$ PRF, or a uniformly random Boolean  function  otherwise. Using the definition of $\widetilde{O}$ and the fact that $\De$ is a
  distinguisher for the One-Bit $\RLWR$ PRF,
  one can conclude that $\De'$ is indeed a distinguisher for the $\RLWE$ degree-$k$ PRF.

  We finish the proof by showing how $\De'$ can simulate $\widetilde{O}$ using queries to $O'$.
  We now describe the action of $\widetilde{O}$ on the query $\ket{x}\ket{y}$ made by $\De'$ -- without loss of generality, it suffices to define the 
  actions of $\widetilde{O}$ the basis states since as the operations of $\De'$ are unitary.   $\De'$ makes a call to $O'$ using a $\ket{x}$ as the input register and fresh
  ancilla bits set to $\ket{0}$ as the output register, resulting in 
  $\ket{x}\ket{y}\ket{z}$, where $z = O'(x)$. $\De'$ computes the state
  $\ket{x}\ket{y \oplus z_1}\ket{z}$, where $z_1$ denotes the first bit of~$z$, followed by a second query 
  to $O'$ using the first and third registers as the input and output registers, respectively.
  $\De'$ then discards the last register (which has value $\ket{0}$) and answers the query
  with $\ket{x}\ket{y \oplus z_1}$.%
  Notice that if $O'$ is an oracle to the $\RLWE$ degree-$k$ PRF family, then $\widetilde{O}$ is simulating a function from the One-Bit $\RLWR$ PRF family; and  if $O'$ is a
  random function, the~first bit of the representation of its output is a random
  bit, which ensures that $\widetilde{O}$ is also a random~function.
\end{proof}

\vspace{-0.8em}
Putting all the security proofs together, we obtain the choice of parameters that would make some One-Bit $\RLWR$ PRF family quantum secure.
\begin{lemma}
\label{lem:qprf}
Let $d \in \natural$, $r > 0$, $2 \leq p \ll q$ and $k =\omega(\log d)$ such that $\log q \leq O(\poly(d))$. 

	Let $R$ be a degree-$d$ cyclotomic ring and let $\Upsilon(r)$ be the
    $d$-dimensional discrete Gaussian distribution over $R$ with standard
    deviation at most $r$ in each coefficient. Let $q \geq p \cdot k(r
    \sqrt{d+k} \cdot \omega(\log d))^k \cdot d^{\omega(1)}$. Let $\RF$ be the
    $\RLWRPRF_{d,q,p,k}$ function family where each secret $\mbf{s}_i$ is
    independently drawn from~$\Upsilon(r)$. 
    {If $\RLWE_{d,q,\Upsilon(r)}$ is hard for an $t(n)$-time quantum
    computer, then $\RF$ is
    a quantum-secure pseudo-random function family against $O(t(n))$
    adversaries.}
\end{lemma}
\begin{sproof}
We begin with the security proof for the $\RLWEPRF_{d,q,p,k}$ against classical adversaries as shown in~\cite[Section~5]{BPR12}. By making adjustments almost identical to those in Zhandry's proof~\cite[Section~F]{zhandry:qprf} that was used to prove the quantum security of the $\LWE$ base PRF, we conclude that the $\RLWEPRF_{d,q,p,k}$ is secure against quantum adversaries. Using Lemma~\ref{lem:prf_sec}, we obtain the parameters $q, p, k$ for which the $\RLWRPRF_{d,q,p,k}$, $\RF$ is secure against quantum adversaries. %
\end{sproof}

Next, we demonstrate that there exist efficient constant-depth circuit implementations for the One-Bit $\RLWR$ PRF family $\RF$.

\begin{lemma}
\label{lem:PRFinTC0}
Let $d \in \natural$, $k = \omega(\log d)$, let $q =d^{\omega(1)}$ and $2 \leq p \leq q$ be powers of $2$ such that $\log q \leq O(\poly(d))$. Let $R$ be the degree-$d$ cyclotomic ring. Consider the $\RLWRPRF_{d,q,p,k}$ function family $\RF$ defined in~\Cref{def:RLWRPRF}. For every $\mbf{a}, \mbf{s}_1, \ldots, \mbf{s}_k \in R_q$, 
$f_{\mbf{a}, \mbf{s}_1, \ldots, \mbf{s}_k} \in \RF$ can be computed by an $O(\poly(d, k))$-sized $\TCz$~circuit. 
\end{lemma}

\begin{proof} 
From~\Cref{def:RLWRPRF}, notice that given a $k$-bit input $x$, for any function
  $f_{\mbf{a}, \mbf{s}_1, \ldots, \mbf{s}_k} \in \RF$, $f_{\mbf{a}, \mbf{s}_1,
  \ldots, \mbf{s}_k} (x)$ can be computed as follows: (a) take the product of
  $\mbf{a}$ with those secrets $\mbf{s}_i$ for which $x_i=1$ to obtain an
  element in $R_q$, (b) round it to an element in $R_p$ and (c) output the first
  bit of the representation of this element. Since $q$ and $p$ are powers of
  $2$, the coefficient-wise rounding procedure just truncates each coefficient to the $\log p$ most significant bits. Hence, the last two steps can be efficiently executed by a $\poly(d)$-sized $\TCz$ circuit as there are most $d$ coefficients to round. It remains to show that computing the iterated multiplication of at most $k+1$ ring elements can be performed using a $\TCz$ circuit. 

Recall that the ring elements, in the standard representation are stored as degree-$d$ polynomials and multiplication in this representation may not be efficient. However, using a canonical embedding $\sigma: R_q \rightarrow \mathbb{C}^d$, arising from algebraic number theory, any ring element $z \in R_q$ can be mapped to the complex vector $(z(\eta_i))_i \in \complex^d$ where $\eta_i$ is the $i$th complex root of $(-1)$. In other words, $z$ can be efficiently stored as complex vectors of length $d$ and each vector entry of size at most $O(\log q) = O(\poly d)$ bits. The advantage in this representation is that ring multiplication reduces to coefficient wise multiplication of these vectors~\cite{LPR10,LPR13}. Now, Step (a) can be broken up into the following operations:
\begin{enumerate}
\item Convert the $k+1$ ring elements ${\mbf{a}, \mbf{s}_1, \ldots, \mbf{s}_k}$ into their canonical embeddings $\sigma(\mbf{a}), \{\sigma(\mbf{s}_i)\}_{i\in [k]} \in~\mathbb{C}^d$.
\item Perform coordinate wise product of at most $k+1$-vectors in this embedding.
\item Convert the solution back into the standard representation, i.e., perform the inverse of the embedding $\sigma$ %
\end{enumerate}

The first and last steps to convert to and from the canonical embedding is
  performed using Fast Fourier Transforms (FFT) and its variants or the Chinese
  Remainder Representation (CRR)~\cite{LPR13} which can be executed in
  $\TCz$~\cite{RT92,HAB02}. The second step is just $d$-parallel iterated
  multiplications of at most $k+1$ elements each of $O(\log q) = O(\poly
  d)$-bits which can be performed by an $O(\poly(d, k))$-sized $\TCz$
  circuit~\cite{HAB02}. Putting all the steps together,  $f_{\mbf{a}, \mbf{s}_1, \ldots, \mbf{s}_k} (x)$ can be computed by an $O(\poly(d, k))$-sized $\TCz$ circuit.
\end{proof}

\vspace{-0.8em}
One problem that prevents us from computing any function from the family $\RF$
in smaller circuit classes, say $\ACz$, lies in the fact that performing FFT,
the CRR embedding and iterated multiplication cannot be done efficiently in
$\ACz$. However, we use two observations to demonstrate that there exist
sub-exponential
sized $\ACz$ circuits that compute functions in the family $\RF$: (i) arithmetic
operations on $O(\polylog n)$-bit numbers can be executed using $\poly(n)$-sized
$\ACz$ circuits\footnote{Basic arithmetic operations can be expressed by
circuits containing \textsf{AND}, \textsf{NOT}, \textsf{OR} and \textsf{MAJ}
gates. \textsf{MAJ} gates with $O(\polylog n)$-bit fan-in can be computed by
$\ACz$ circuits.}~\cite{MT98}; (ii) iterated multiplication of $t$ numbers each
of length at most $n$ bits can be performed in $\poly(2^{n^\eps},
2^{t^\eps})$-sized $\ACz$ circuits~\cite{HAB02,HV06}. We proceed by 
defining the PRF with a {smaller dimension in the underlying problem
such that the PRF's circuit size scales sub-exponentially in the dimension}.
This resembles the technique used by Kharitonov~\cite{Kharitonov93} to show that some instances  of the BBS PRG can be implemented in $\ACz$. 

\begin{lemma}
\label{lem:PRFinAC0}
Let $d \in \natural$ and choose constants $c > \eta > 1$ and $\eps < 1$ such that $\eps \leq \frac{\eta}{\eta+c}$. 
Let $d > 2$ and $k = (d)^{\frac{1}{\eta}}$. Choose $2 < p \ll q = \superpoly(d)$ as powers of
  $2$ such that $\log q = \frac{k^2}{k+1} \log d$. Consider the $\RLWRPRF_{d,q,p,k}$
  family $\RF$ defined in~\Cref{def:RLWRPRF}. For every $\mbf{a}, \mbf{s}_1,
  \ldots, \mbf{s}_{k} \in R_{q}$, 
  $f_{\mbf{a}, \mbf{s}_1, \ldots,
  \mbf{s}_{k}} \in \RF$  can be computed by a $2^{O(d^{1/c})}$-sized $\ACz$
  circuit of depth $O(1/\eps)$. 
\end{lemma}
\begin{proof}
For notational simplicity, we set $d' = 2^{(d)^{1/c}}$ i.e., $d = \log^c d'$ and proceed by first calculating the circuit size in terms of $d'$. 
The proof follows the ideas used in the proof of~\Cref{lem:PRFinTC0}.
  Specifically, the operations performed to evaluate $f_{\mbf{a}, \mbf{s}_1,
  \ldots, \mbf{s}_{k}}(x)$ break up in the same manner as before. Also, note
  that the input to the circuit is determined by the $k$-bit string $x$.

  Let us
  start our analysis with the 
  rounding of some element from $R_{q}$ to $R_{p}$. This can be done by
  truncating each coefficient to the most significant $\log p$ bits where $\log
  p \ll \log q = \frac{k^2}{k+1} \log d \leq  \ k \log d = O(\polylog d')$  since $d = \log^c d'$ and $k = (d)^{1/\eta} =
  { \log^{c/\eta} d'}$. From observation (i) above, this truncation can be performed by a $\poly(d')$-sized $\ACz$ circuit for each of the $d$ coefficients.

  In order to deal with the iterated multiplication of $\mbf{a} \in R_{q}$ with
  the secrets $\mbf{s}_1, \ldots, \mbf{s}_{k} \in R_{q}$, as before we first convert the ring elements from their standard representation to the canonical embedding $\sigma : R_{q} \rightarrow \complex^{d}$ so that every element $z \in R_{q}$ is a complex vector of length $d = O(\polylog d')$ and each vector entry of size at most $O(\log q) = O(\polylog d')$ bits~\cite{LPR13}. The complexity of the FFT and CRR operations needed to convert to and from the canonical embedding $\sigma$ in this case are constrained by the complexity of multiplication and exponentiation of $O(\log q)$-bit numbers. 
Then, using a variant of observation (ii) above that also works for
  exponentiation~\cite{HAB02, HV06}, these operations can be performed by
  $\poly\left({2^{(\log q)}}^\eps, {2^{k}}^\eps\right)$-sized $\ACz$ circuits of
  depth $O(1/\eps)$. Finally, for the $d$-parallel iterated multiplications of
  at most $k+1$ elements each of $O(\log q)$-bits, by observation (ii) we need $\poly\left({2^{(\log q)}}^\eps, 2^{k^\eps}\right)$-sized $\ACz$ circuits of depth $O(1/\eps)$ to evaluate it. 

The total circuit size for evaluating $f_{\mbf{a}, \mbf{s}_1, \ldots,
  \mbf{s}_{k}}(x)$ hinges on the complexity of $\poly({2^{(\log q)}}^\eps,
  2^{k^\eps})$ which we show to be $O(\poly(d'))$ below. First, notice that
  $2^{k^\eps} \leq 2^{{(\frac{k}{2} \log d)}^\eps} \leq  2^{{(\frac{k^2}{k+1} \log d)}^\eps} = {2^{(\log q)}}^\eps$ for $d > 4$.
  Hence, it suffices to bound only the complexity of the latter term in terms of
  $d'$:
\begin{align*}
	(\log q)^\eps & = \left(\frac{k^2}{k+1} \log d \right)^\eps \leq (k \log d)^\eps = (d)^{\frac{\eps}{\eta}} (\log d)^{\eps} & \left( \text{Since, } k = (d)^{\frac{1}{\eta}} \right)\\
	& =  (\log d')^{\frac{c \eps}{\eta}} O( (\log\log d')^\eps ) & \left( \text{Since, } d = \log^c d' \right)\\
	& \leq (\log d')^{\frac{c \eps}{\eta}} O((\log d')^\eps) & \\
	& = O\left( (\log d')^{\left(\frac{c}{\eta} + 1 \right) \eps} \right) & \\
	& = O\left( (\log d')^{\frac{c + \eta}{\eta} \eps} \right) & \\
  & = O(\log d') & \left(\text{Since, } \eps \leq \frac{\eta}{\eta + c} \right),
\end{align*} 
which implies that 	${2^{(\log q)}}^\eps  \leq 2^{O(\log d')} = O(\poly(d'))$.

Putting all the steps together, the function $f_{\mbf{a}, \mbf{s}_1, \ldots, \mbf{s}_{k}}(x)$ taking as input a $k$-bit number $x$ can be computed by an $\ACz$ circuit of size $O(\poly(d')) = O\left(\poly\left(2^{d^{1/c}}\right)\right) = 2^{O(d^{1/c})}$ and depth scaling as~$O(1/\eps)$.
\end{proof}

\vspace{-1em}
\section{Quantum-secure PRFs vs. quantum learning}
\label{sec:quantumsafePRFvslearn}
In this section we prove our main theorem which shows the connection between
efficient quantum learning and quantum algorithms which serve as distinguishers
for pseudo-random functions.  
The structure of our proof goes in the same lines
of the work of Kharitonov~\cite{Kharitonov93,Kharitonov:AC1}: we assume the existence of an efficient learner for
some PRF $\Fe$ and utilize it to construct an efficient distinguisher for $\Fe$ from
truly random functions. However,
we need here very different techniques, since the arguments used by
Kharitonov do not follow in the quantum setting.

A key part of the proof for this theorem hinges on bounding how well a quantum
algorithm predicts the output of a random function on a random input.
We assume that the algorithm is allowed to query an oracle to the function
(in superposition) at most $\query$ times. Each of these queries denotes performing a membership
query to $z$, the truth-table of the function. We show an upper bound on the
probability of predicting the output of this random function for a uniformly
random input.

\begin{lemma}
\label{lem:upperboundlemma}
  Let $k > 1$ and $f : \01^k \to \01$ be a random function. Consider a quantum
	  {algorithm} that makes $\query$ quantum membership queries to $f$. Given a uniformly random question $x\in \01^{k}$,
  the probability of the quantum algorithm correctly predicts $f(x)$ is    \begin{align}
    \label{eq:ub-probability}
    \Pr_{x\in \01^{k}}[h_x = f(x)] \leq  \frac{1}{2}+
    \sqrt{\frac{k\cdot \query}{2^k}},
  \end{align}
where $h_{x}$ is the output of the algorithm given the question $x$.
\end{lemma}
\begin{proof}
Let us denote $\gamma(k) =
  \sqrt{\frac{k \cdot \query}{ 2^{k}}}$. On input some index $x$, we view $h_x$ as the output of a ``hypothesis function" $h:\01^k\rightarrow \01$, such that $h_x = h(x)$.
By contradiction, let us assume that
	the hypothesis function $h$ satisfies %
		\begin{align}
	\label{eq:learningcontradiction}
      \Pr_{x \in \01^k}\left[f(x) = h(x)\right] >
      \frac{1}{2} + \frac{\gamma(k)}{2},
	\end{align} 
	where $x$ is drawn uniformly from $\01^k$.  In this case,  the mutual
  information between the truth table of $f$, which we denote as $z$,  and the truth-table of the
  hypothesis $h$  is given by
	\begin{align}\label{eq:lb-information-entropy}
	\mathbf{I}(z:h) = \sum_{x\in \01^k} \mathbf{I}(z_x:h) = \sum_{x\in \01^k}
    (\mathbf{H}(z_x) - \mathbf{H}(z_x|h)) = 2^k  - \sum_{x\in \01^k} \mathbf{H}(z_x|h),
	\end{align}
	where the first and last equality used the independence of the $z_x$s because
  the truth table of $f$, i.e., $z\in \01^{2^k}$, is a uniformly random string. By Fano's inequality (in \Cref{lem:fano}), it follows that $\mathbf{H}(z_x|h) \leq \mathbf{H}_b(p_x)$, where $\mathbf{H}_b(\cdot)$ is the binary entropy function and $p_x$ is
	the probability of error on guessing $z_x$ by an arbitrary estimator whose
  input is $h:\01^k\rightarrow \01$. By assumption of $h$ in
  \Cref{eq:learningcontradiction}, it follows that $\frac{1}{2^k}\sum_x
  p_x < \frac{1}{2} - \frac{\gamma(k)}{2}$. Using this, we then have that 
	$$
	\sum_x \mathbf{H}(z_x|h) \leq \max_{\{p_x\}}
	\sum_x  \mathbf{H}_b(p_x),
	$$
	where the maximization is over $\{p_x\}$ in the set $\left\lbrace p_x :
    \frac{1}{2^k}\sum_x p_x \leq \frac{1}{2} - \frac{\gamma(k)}{2}\right\rbrace$.
	Since $\mathbf{H}_b$ is a concave function, the maximum is obtained when all the $p_x$s are equal (this also follows from Jensen's inequality). Given our upper-bound on the sum $\sum_x p_x$, it follows that $	\sum_x  \mathbf{H}_b(p_x)$ is maximized when $p_x =
  \frac{1}{2} - \frac{\gamma(k)}{2}$ for every $x$. In this case, we have 
	\begin{align}
	\label{eq:ub-entropy}
\max_{\{p_x\}}
\sum_x  \mathbf{H}_b(p_x)\leq    \sum_x \mathbf{H}_b\left(\frac{1}{2} -
    \frac{\gamma(k)}{2}\right)=  2^k \cdot \mathbf{H}_b\left(\frac{1}{2}
    - \frac{\gamma(k)}{2}\right) \leq 2^k \left(1 -
    \frac{2\gamma(k)^2}{\ln{2}} \right),
	\end{align}

	where we use \Cref{fact:taylorseriesbinaryentropy} in the last inequality. Putting together \Cref{eq:lb-information-entropy} and~\eqref{eq:ub-entropy}, we~obtain
	\begin{align}\label{eq:lb-information}
    \mathbf{I}(z:h) >  \frac{2^{k+1} \gamma(k)^2}{\ln{2}} > 2 \cdot k \cdot \query.
	\end{align}
	
	We now upper bound the mutual information between the output hypothesis $h$
  and the uniformly random string $z$. One way to view the quantum 
  algorithm is as a protocol where a quantum membership query to $z$ is a
  message from the algorithm to an oracle hiding $z$ and the oracle's output is a message from
  the oracle to the algorithm. In this case, 
    using \Cref{cor:communicationmutualinformation}
    the mutual information between the
  output hypothesis $h$ and the random string $z$ can be upper-bounded by the
  communication complexity of the protocol, %
	\begin{align}\label{eq:ub-information}
    \mathbf{I}(h : z) \leq (k+1)\cdot \query.
	\end{align}

  However, since $k \geq 2$, the lower bound in
    \Cref{eq:lb-information} contradicts \Cref{eq:ub-information},
    which in turn contradicts our assumption in
    \Cref{eq:learningcontradiction}. 
\end{proof}

\vspace{-0.8em}

Now, we can prove the main technical result of this section.
\begin{theorem}
	\label{thm:distinguisher-prf}
  Let $n,D\geq 2$ and $k= D \cdot \log n$. Let $\Fe = \{f_K : \01^k \rightarrow \01 \ | \ K \in~\01^n \}$
  be a family of functions. Suppose there exists a $t(n)$-time uniform quantum-PAC learner for~$\Fe$ (given
  access to $\query$ quantum membership queries and uniform quantum examples)
  with bias 
 {  $\biaslearning(n) \geq 
    2\sqrt{\frac{k\cdot \query}{n^D}}$}

  \noindent
  Then there exists an
  $O(t(n))$-time quantum algorithm
  {$\De$}
  that satisfies:
  \begin{align}
	\label{eq:distinguisher}
    \Big \vert \Pr_{f\sim \Fe} [\De^{\ket{f}}(\cdot)=1]- \Pr_{f \sim \mathcal{U}}
    [\De^{\ket{f}}(\cdot)=1]\Big| \geq \frac{\biaslearning(n)}{2},
	\end{align}
  where the probability is taken uniformly over $f\in \Fe$ and $f \in \mathcal{U}$,
  where $\mathcal{U}$ is the set of all functions from $\01^k$ to $\01$. 
\end{theorem} 
	\begin{proof}
  Let $\A$ be a $t(n)$-time uniform quantum-PAC learner that makes at most~$\query$ quantum queries to a concept $c\in \Cc$ and outputs a hypothesis $h:\01^{k}\rightarrow \01$ such that
	$$
  \Pr_{x\in \01^k}\left[c(x) = h(x)\right] \geq \frac{1}{2} + \beta(n),
	$$
	where the probability is over $x$ drawn uniformly from $\01^k$. Note that $\A$ can obtain a uniform quantum example by making a single quantum membership query, so  without loss of generality we assume $\A$ makes membership queries.
	
  The goal is to use $\A$ to construct a quantum
  distinguisher $\De$ for the PRF $\Fe$ that
  satisfies \Cref{eq:distinguisher}.
    Let the distinguisher
    $\mathcal{D}$ have quantum oracle access to the function $f : \01^{k} \to \01$. The goal for $\mathcal{D}$ is to
    decide if $f$ is a uniformly random function or if $f \in \Fe$. In order to
    do this, $\mathcal{D}$ proceeds by first running the
    quantum learning algorithm $\mathcal{A}$ as follows: whenever $\mathcal{A}$
    makes a quantum membership query, $\mathcal{D}$ uses its oracle to answer
    $\mathcal{A}$.

  After making $\query$
  membership queries, $\mathcal{A}$ then outputs a hypothesis $h$. The
  distinguisher $\De$ outputs~$1$ (i.e., $f$ is a pseudo-random function) if and
    only if $h(x) = f(x)$ for a uniformly random $x \in \01^k$. 

    The running time of $\De$ is $O(t(n))$ since $\A$
    runs in time $t(n)$, each query has constant cost and computing $h(x)$ takes
    at most $O(t(n))$ time. 
    We now show that $\De$ satisfies \Cref{eq:distinguisher}. Suppose $f$ is
    a pseudo-random function i.e., $f \in
    \Fe$. Each query from $\mathcal{A}$ will be answered
    correctly as $\De$ has oracle access to it.
  Hence, we have
	\begin{align}\label{eq:lb-pseudorandom}
    \Pr[\mathcal{D}^{\ket{f}}(\cdot) = 1] =
	\Pr_{x \in \01^k}\left[f(x) = h(x)\right] \geq
    \frac{1}{2}+\biaslearning(n), 
	\end{align}
	where the inequality follows from the correctness of the quantum learning
  algorithm $\A$.

  On the other hand, if $f$ was a uniformly random function, using \Cref{lem:upperboundlemma}, we can conclude~that 
	\begin{align*}
    {
      \Pr[\mathcal{D}^{\ket{f}}(\cdot) = 1]} =
	\Pr_{x \in \01^k}\left[f(x) = h(x)\right] \leq
    \frac{1}{2}+
    \sqrt{\frac{k\cdot \query}{2^{k}}}=  \frac{1}{2}+
    \sqrt{\frac{k\cdot \query}{n^D}}
	\end{align*}
	where the last equality uses the definition of $k=D\cdot \log n$.

Combining these two cases, the bias of the distinguisher $\De$ is

\begin{align*}
    \Big \vert \Pr_{f \in \Fe} [\De^{\ket{f}}(\cdot)=1]- \Pr_{f \in \mathcal{U}}
  [\De^{\ket{f}}(\cdot)=1]\Big| &\geq \Big(\frac{1}{2}+\biaslearning(n) \Big)- \Big(\frac{1}{2}+
    \sqrt{\frac{k \cdot \query}{n^D}} \Big)
     = {\beta(n)}-\sqrt{\frac{k \cdot \query}{n^D}}  \geq \frac{\biaslearning(n)}{2},
\end{align*}

 where  the last inequality uses the definition of $\beta(n)$.
This concludes the proof.
\end{proof}

Our hardness results for $\TCz$ and $\ACz$ are based of the following corollary. In order for notational simplicity in the subsequent sections, we replace $n$ in Theorem~\ref{thm:distinguisher-prf} by $M$ below.

\begin{corollary}
	\label{cor:quantumsafePRF}
  Let $M\geq 2$, $D \in \{2,\ldots,
  M/(4\log M)\}$ and $a>0$ be a constant. Let $\Fe = \{f_K : \01^k \rightarrow \01 \ | \ K \in~\01^M \}$
  be a family of functions for some $k = D \cdot \log M$. %

For any $t(M)\leq M^{D-a}$ and 
{$\beta(M)\geq 2 \sqrt{\frac{k}{M^a}}$}, if there exists a $t(M)$-time uniform
  quantum-PAC learner for~$\Fe$ with quantum membership queries 
	and bias  $\biaslearning(M)$, then there exists a distinguisher that runs in
  time $O(t(M))$ and is able to distinguish $\Fe$ from a uniform random
  function.
\end{corollary} 
\begin{proof}
	First observe that the number of queries made by the learner is at most
  $\query\leq M^{D-a}$. Hence, 
{
    \begin{align*}
  \beta(M)\geq 2\sqrt{\frac{k}{M^a}}
  =  2\sqrt{\frac{k M^{D - a}}{M^D}} \geq
  2\sqrt{\frac{k \cdot \query}{M^D}},
    \end{align*}}
	which satisfies the assumption of Theorem~\ref{thm:distinguisher-prf}. Hence,
  there is a $O(t(M))$-time quantum distinguisher for the PRF family $\Fe$ with
  distinguishing probability $\beta(M)=1/\poly(M)$ since we let $a>0$ be a constant. 
\end{proof}

\subsection{Hardness of learning \torp{$\TCz$}{TC0}}
\label{sec:TC0prghardness}
For completeness, we formally state \Cref{result:tc0} and follow with its proof. 

\begin{theorem}
	\label{thm:hardnessoftc0}
  Let $n \in \natural$, $q = 2^{O(n^c)}$ for some constant $c\in(0,1)$, $r \leq \poly(n)$ and $t(n)\leq 2^{O(\polylog (n))}$. Let $\Upsilon(r)$ be the $n$-dimensional discrete Gaussian distribution over a degree-$n$ cyclotomic ring with standard deviation at most $r$ in each coefficient. Assuming
  that the $\RLWE_{n,q,\Upsilon(r)}$ problem is hard for an $O(t(n))$-time quantum
  computer, then there is no $O(t(n))$-time uniform weak quantum-PAC learner for  {$\poly(n)$-sized}~$\TCz$ circuits.
\end{theorem}

\begin{proof}
	We proceed by a proof of contradiction. First, consider the concept class to
  be the $\RLWRPRF_{n,q,p,k}$ function family $\RF$ as given by
  Definition~\ref{def:RLWRPRF}, where $2 \leq p \ll q$ are powers of~$2$,
  $k = \log^\alpha n = \omega(\log n)$ for some constant $\alpha > 2$, {$p = 2$}
  and $q = 2^{\gamma \cdot \frac{k}{k+1}}$  and $\gamma = n^c$ for $0
  < c < 1$. Using Lemma~\ref{lem:PRFinTC0} for $d = n$, every function in $\RF$ can be computed by an ${ O(\poly (n), \omega(\log n))} = O(\poly(n))$-sized $\TCz$ circuit. In other words, the concept class $\RF \subseteq \TCz$. 

	By way of contradiction, suppose that there is a $t(n)$-time uniform weak quantum learning algorithm for $\RF$, i.e., 
  there exists a learner $\A$ for $\RF$ that uses at most $t(n)$ queries/samples and
	achieves bias $\biaslearning(n)=n^{-\delta}$ for some constant
  $\delta>0$. 

  Let $M = (k+1)\log{q} =\gamma k= \gamma \log^{\alpha}n$ (where the first equality used that $\log q=\gamma\cdot k/(k+1)$ by definition) and $D = \frac{k}{\log M} = 
  \frac{\log^{\alpha}n}{\log \gamma +\alpha\log\log n}$.
 Let $a$ be some constant satisfying $a\geq 2(\delta+1)/c$. This choice of $a$ allows us to show that  
$$
\beta(n)\geq \frac{1}{n^\delta}\geq \frac{1}{n^{ca/2-1}}\geq 2\sqrt{\frac{(\log n)^{\alpha(1-a)}}{n^{ca}}}=2\sqrt{\frac{(\log n)^{\alpha(1-a)}}{\gamma^a}} =2\sqrt{\frac{\log^\alpha n}{(\gamma\log^{\alpha } n)^a}} = 2\sqrt{\frac{k}{M^a}},
$$
 where the last equality uses the definition of $M$ and the penultimate equality uses the definition of $\gamma$.
  Finally, notice that 
  $$
  D - a =  \frac{\log^{\alpha}n}{\log \gamma +\alpha\log\log n}-a=\frac{\log^{\alpha}n}{c\log n +\alpha\log\log n}-a\geq \frac{\log^{\alpha-1}n}{2c},
  $$
  using $\gamma=n^c$ in the second equality. Therefore we have
  $$ M^{D - a} \geq  \left(\gamma \log^{\alpha}n\right)^{\frac{\log^{\alpha-1}n}{2c}}
  = 2^{O(\polylog (n))}\geq t(n).
  $$
This choice of parameters allows us to use Corollary~\ref{cor:quantumsafePRF},
which 
implies the existence of an $O(t(n))$-time
  distinguisher for the $\RLWRPRF_{n,q,p,k}$ family $\RF$. For the
  aforementioned choices of $n, q$ and~$r$, using Lemma~\ref{lem:qprf},
  we can translate the algorithm for $\RF$ to an $O(t(n))$-time quantum distinguisher for the $\RLWE_{n,q,\Upsilon}$ problem, which contradicts our assumption.
\end{proof}
An immediate corollary of~\Cref{thm:hardnessoftc0} using Assumptions
\ref{assumption:polynomial-hardness} and \ref{assumption:quasipoly-hardness} gives us our main results:

\begin{corollary}\label{cor:tc0-hard}
  Let $n \in \natural$, $q = 2^{O(n^c)}$ for some constant $0 < c < 1$, $r \leq \poly(n)$ and $\Upsilon(r)$ be discrete Gaussian distribution over degree-$n$ cyclotomic rings with standard deviation at most $r$ in each coefficient.
	\begin{enumerate}
		\item  Assuming that the $\RLWE_{n,q,\Upsilon(r)}$ problem is hard for a $\poly(n)$-time quantum computer, then there is no $\poly(n)$-time uniform weak quantum-PAC learner for  {$\poly(n)$-sized} $\TCz$ circuits.
		\item  Assuming that the $\RLWE_{n,q,\Upsilon(r)}$ problem is hard for a $2^{O(\polylog(n))}$-time quantum computer, then there is no $2^{O(\polylog(n))}$-time uniform weak quantum-PAC learner for  {$\poly(n)$-sized} $\TCz$ circuits.	\end{enumerate}
\end{corollary}

\subsection{Hardness of learning \torp{$\ACz$}{AC0}}
\label{sec:AC0prghardness}

The proof for the hardness of learning $\ACz$ is similar to the proof in the
previous section and will proceed by showing that for a suitable choice of
parameters, the concept class constructed from $\RF$ can be evaluated by $\ACz$
circuits. Recall from~\Cref{lem:PRFinAC0}, that the circuit size scales
sub-exponentially with respect to the {dimension of the underlying problem}.
Hence, in order to prove the $\ACz$ hardness, we make two changes. We first
consider ``smaller instances" of the $\RLWRPRF$ that can now be computed in
$\ACz$ (instead of $\TCz$ as we showed in \Cref{lem:PRFinTC0}). Secondly, we
``weaken" the statement of our hardness result by showing that, the existence of
a quasi-polynomial  time \emph{strong} quantum learner for $\ACz$ implies the existence of a strongly sub-exponential time quantum distinguisher for $\RLWE$ (instead of polynomial-time quantum distinguishers for $\RLWE$ as we showed in \Cref{thm:hardnessoftc0}). 

We first give reasons as to why the strongly sub-exponential time assumption is
justified. The best known algorithms for $d$-dimensional $\LWE$ using various
methods from lattice reduction techniques~\cite{LLL82, CN11}, combinatorial
techniques~\cite{BKW03,Wag02} or algebraic ones~\cite{AG11} all require
$2^{O(d)}$-time in the asymptotic case. Even by quantizing any techniques, the
improvements so far only affect the constants in the exponent and do not provide
general sub-exponential time algorithms for $\LWE$. It is also believed that the best-known algorithms for Ring-$\LWE$ do not do much better than the lattice reduction/combinatorial methods used for $\LWE$. While the exponential run-times correspond to the hardest instances of these problems, for the choice of parameters used in this section, the algorithms scale more like $2^{\Omega(\sqrt{d})}$. Hence, we are justified in \emph{weakening} our statement of hardness to the assumption that $\RLWE$ does not have $2^{d^{1/\eta}}$-time algorithms for every constant $\eta > 2$.

Before going into the proof, we provide some intuition as to why the choice of parameters as defined in \Cref{lem:PRFinAC0} forces us to impose stronger hardness assumptions  -- the circuits are $\poly(n)$-sized and the algorithms run in $\poly(n)$-time but, the family $\RF$ is defined on a ring of smaller dimension i.e., $d = O(\polylog n)$. However, from the perspective of the {input} size which scales as $\poly(d)$, this means that the algorithms scale sub-exponentially in $d$. Hence, we need to use the more stringent ``no strongly sub-exponential time algorithms for $\RLWE$" assumption. We formally restate \Cref{result:ac0} below and proceed with its proof.

\begin{theorem}
\label{thm:hardnessofac0}
Let $c > \eta > 2$ such that $c/\eta>2$, and  $a\geq 2$ be constants. Let $n \in
  \natural$, $d = \log ^c n$, $k = d^{1/\eta}$, $q = d^{O(k)}$ and $r \leq
  d^\lambda$ for $0 < \lambda < 1/2 - O(1/k)$, $\Upsilon(r)$ be discrete Gaussian distribution over degree-$n$ cyclotomic rings with standard deviation at most $r$ in each coefficient. Let $t(d) \leq 2^{d^{1/\eta}}$ and {$\beta(n) \geq 2{(\log n)^{c(1-2a)/2\eta}}$}.  Assuming that $\RLWE_{d,q,\Upsilon(r)}$ is hard for a $t(d)$-time quantum computer,  there is no $2^{(\log n)^{c/\eta}}$-time  uniform quantum-PAC learner {with bias $\beta(n)$} for  $\poly(n)$-sized $\ACz$ circuits. %
\end{theorem}
\begin{proof}
Consider the constants $c, \eta, \eps$ such that $c> \eta > 2$ and $\eps \leq
  \frac{\eta}{\eta+c}$. Consider the concept class $\RF$ from the function
  family $\RLWRPRF_{d,q,p,k}$  as given by~\Cref{def:RLWRPRF} where dimension $d
  = \log^c n$, input $k = d^{\frac{1}{\eta}}$, and moduli $2 < p \ll q$
  are powers of $2$ such that $\log q = \frac{k^2 \log{d}}{k+1}$.
  By~\Cref{lem:PRFinAC0}, every function in $\RF$ can be computed by a $2^{O\left( d^{1/c} \right)}$-sized $\ACz$ circuit of depth $O(1/\eps)$ which is $O(1)$ since $\eps\leq \eta/(\eta+c)\leq 1$. In other words, the $\ACz$ circuits are of size $2^{O\left( (\log n)^{c \cdot 1/c} \right)} = O(\poly(n))$.  

We now consider a function class $\RF$ against which we are going to prove our
  hardness result: keeping Lemma~\ref{lem:PRFinAC0} in mind, recall that the
  length of the key (i.e., {number of bits to define the parameters $\mbf{a},
  \mbf{s}_i\in R_q$ of the function $f_{\mbf{a},
  \mbf{s}_1, \ldots, \mbf{s}_{k}}$}) is at most
  {$M:= (k+1)\log q= d^{2/\eta}\log{d}$}, let
$$
  \RF = \{f_K : \01^k \rightarrow \01 \ | \ K \in~\01^{{(k+1)\log q}} \}.
$$ 

In order to discuss the hardness of learning the function class $\RF$ in terms
  of the hardness of $d$-dimensional $\RLWE$, we now relabel the function class
  above as follows: since {$(k+1) \log q = k^2 \log d$} and $k=d^{1/\eta}$, we have:
$$
\RF = \{f_K : \01^k \rightarrow \01 \ | \ K \in~\01^{d^{2/\eta}\log d} \}. 
$$

By way of contradiction, suppose that there exists a $t(d)\leq
  2^{d^{1/\eta}}=2^{(\log n)^{c/\eta}}$-time uniform  quantum learner, {with
  bias $\beta(n)$}, for $\poly(n)$-sized $\ACz$ circuits i.e., there exists a
  learner $\A$ for $\RF$ that uses at most $t(d)$ queries/samples and achieves
  bias $\beta(n)$. 

  Let $D =\frac{\sqrt{M}}{\log M}$. Observe that $k = D \cdot
  \log{M}$,
{
\begin{align} \label{eq:betadintermsofn}
  \beta(n)\geq 2{(\log n)^{c(1-2a)/2\eta}}=2 {d^{\frac{1}{2\eta}-\frac{a}{\eta} }}=
  2\sqrt{\frac{d^{1/\eta}}{d^{2a/\eta}}}\geq 
  2\sqrt{\frac{k}{M^{a}}}.
\end{align}}

Additionally, observe that 
$$
(D-a)\log M=\Big(\frac{\sqrt{M} }{\log M} -a\Big)\cdot \log M \geq \sqrt{M}/2 \geq d^{1/\eta}(\sqrt{\log d})/2 \geq d^{1/\eta}
$$

 Hence, it follows that
  $M^{D-a}\geq 2^{d^{1/\eta}}\geq  t(d)$. %
We are now ready to apply Corollary~\ref{cor:quantumsafePRF} since $M, k, a$ and
  $D$ satisfy the required properties. It now follows from
  Corollary~\ref{cor:quantumsafePRF} that %
there exists an $O(t(d))$ time distinguisher for the
  $\RLWRPRF_{d,q,p,k}$ problem with distinguishing advantage at least
  $\beta(n)/2 \geq  O(1/\poly(d))$, from \Cref{eq:betadintermsofn}.
  For the aforementioned choices of $d,q$ and $r$, using~\Cref{lem:qprf},
  this distinguisher can be converted into an $O(t(d))$ time algorithm for the
  $\RLWE_{d,q,\Upsilon(r)}$ problem with an $O(1/\poly(d))$ distinguishing
  advantage. This contradicts the assumption that algorithms for
  $\RLWE_{d,q,\Upsilon(r)}$ require $2^{\Omega(d^{1/\eta})}$-time and concludes the
  proof. 
\end{proof}
{While the previous theorem lower bounds shows that there is no
quasi-polynomial-time algorithm to learn $\ACz$ circuits with $\beta(n) \geq
O(1/\polylog n)$, we can weaken our assumption to also consider \emph{strong}
learners with bias  $\beta(n) \geq O(1)$ (say, $\beta = 1/6$). Setting $\beta(n)
= 1/6$ in~\Cref{thm:hardnessofac0}, we would obtain a distinguisher for
$d$-dimensional $\RLWE$ that runs in time $2^{O(d^{1/c})}$ with a distinguishing
advantage $\beta(n) = 1/12 \geq 1/\poly(d)$ thereby contradicting the
sub-exponential time hardness of $\RLWE$. Using
Assumption~\ref{assumption:subexp-hardness}, this leads to the following corollary which is the main result in this section.} 

\begin{corollary}\label{cor:ac0-hard}
Let $c > \eta > 2$, $n \in \natural$, $d = \log^c n$, $k = d^{1/\eta}, q = d^{O(k)}$ and $r \leq d^\lambda$ for $0 < \lambda < 1/2 - O(1/k)$. Let $\Upsilon(r)$ be the discrete Gaussian distribution over a degree-$n$ cyclotomic ring with standard deviation at most $r$ in each coefficient. Assuming that the $\RLWE_{d,q,\Upsilon(r)}$ problem is hard for a $2^{o(d^{1/\eta})}$-time quantum computer, then for all $\nu \leq c/\eta - 1$ there exists no $n^{O(\log^{\nu} n)}$-time {strong} uniform quantum-PAC learner can learn $\poly(n)$-sized $\ACz$ circuits. 
\end{corollary}

Notice that this matches the $\Omega(n^{\polylog n})$ lower bound for
\emph{classical} learning $\ACz$  as obtained by Kharitonov \cite{Kharitonov93}
under the assumption factoring is quasi-polynomial-hard. In particular, this
lower bound also matches the $O(n^{\log ^\alpha n})$ sample and time complexity
upper bound for learning $\ACz$ circuits as demonstrated by Linial, Mansour and
Nisan~\cite{LMN93} for some universal constant $\alpha>1$.  Hence, we can
conclude that uniform \emph{strong} quantum-PAC learners cannot offer a
significant advantage for learning $\ACz$ circuits.

{Using PRFs based on $\RLWR$ and our proof technique, we believe that it may not be possible to find a lower bounds for weak quantum learners for $\ACz$ circuits. The reason is two-fold: (1) The circuit size of the PRF family scales as sub-exponential in the dimension $d$ while the key-length scales as $\poly(d) = O(\polylog n)$ for $\ACz$ circuits and (2) this limits the bias $\beta(n)$ to be bounded from below by $O(1/\poly(d)) = O(1/\polylog n)$ while weak learners are expected to have a bias $\beta(n) \geq O(1/\poly(n))$}. We leave the question of whether a different choice of quantum-secure PRFs with a more efficient circuit implementation could improve the lower bound open for future work.

\vspace{-0.5em}
\section{Quantum-secure encryption vs. quantum learning}
\label{sec:quantum-safelearning}
  In this section, we prove a theorem relating the
  security of quantum public-key cryptosystems and the hardness of quantum
  learning the class of decryption functions in the cryptosystem.

Kearns and Valiant~\cite{KearnsV94} showed the  simple, yet powerful
connection between learning and public-key cryptography: consider a
\emph{secure} public-key cryptosystem and define a concept class $\Cc$ as the
set of all decryption functions (one for every public and private key). Assume that there was an \emph{efficient} PAC learner for $\Cc$. Given a set of
encryptions (which should be viewed as labeled examples), supposing an
efficient learner for $\Cc$ was able to learn an unknown decryption function
with non-negligible advantage, then this learner would break the cryptosystem.
This contradicts the assumption that the cryptosystem was secure and, in turn,
shows the infeasibility of learning $\Cc$. Kearns and Valiant used this
connection to give hardness results for learning polynomial-sized formulas,
deterministic finite automata, and constant-depth threshold circuits based on
the assumption that Blum integers are hard to factor or that it is hard to
invert the RSA function. Our main contribution in this section is to quantize the theorem by Kearns and Valiant~\cite{KearnsV94} and draw a relation between quantum learning and security of quantum cryptosystems.

	\begin{theorem}
		\label{thm:cryptotolearn}
		Consider a public-key cryptosystem which encrypts bits by $n$-bit strings.
    Suppose that the (randomized) encryption function is given by
    $e_{\Kpub,r}:\01\rightarrow \01^n$ (where $K=(\Kpub,\Kpriv)$ consists of the
    public and private keys and $r$ is a random string) and the decryption function
    is given by {$d_{\Kpriv}:\01^n\rightarrow \01$}. Let
    $\Cc=\{d_{\Kpriv}:\01^n\rightarrow \01\}_{\Kpriv}$ be a concept class. Let 
		\begin{align*}
		\eps =\Pr_{K,r}[d_{\Kpriv}(e_{\Kpub,r}(0))\neq 0 \text{ or }
		d_{\Kpriv}(e_{\Kpub,r}(1))\neq 1],
		\end{align*}
    be the {negligible} probability of decryption error (where the probability is taken over
    uniformly random $K=(\Kpub,\Kpriv)$ and $r$). If there exists a weak quantum-PAC
 learner for $\Cc$ in time $t(n)\geq n$ that satisfies
    $t(n)\biaslearning(n)=1/n^{\omega(1)}$, then there exists a $t(n)$-time quantum algorithm $\Adv$ that satisfies
		\begin{align}
		\label{eq:quantumdistinguisher}
		\Big\vert \Pr_{K,r^*,b^*}[\Adv(\Kpub,e_{\Kpub,r^*}(b^*))=b^*] -
		\Pr_{K,r^*, b^*}[\Adv(\Kpub,e_{\Kpub,r^*}(b^*))=\overline{b^*}]\Big\vert\geq \frac{1}{\poly(n)},
		\end{align}
		where the probability is taken over uniformly random $K,r^*, b^*$ over their respective domains.
	\end{theorem}
	
	\begin{proof}
		For notational simplicity, for a fixed $\Kpub$, let $e_{r,b} :=
		\mathsf{Enc}_{\Kpub, r}(b)$ and  $\mathcal{E} = \{e_{r,b}:r,b\}$ be the set of all
		encryptions for the public key $\Kpub$. Our goal here is to devise a quantum algorithm
    $\Adv$ that satisfies the following: 
		on input $(\Kpub,e^*)$, where $e^* = e_{r^*,b^*}$ is uniformly randomly chosen from
		$\mathcal{E}$, $\Adv$ outputs $\hat{b}$ such that $\hat{b} = b^*$ 
		with  $1/\poly(n)$ advantage over a random guess. $\Adv$ would clearly serve
    as our quantum algorithm that satisfies \Cref{eq:quantumdistinguisher}, thereby proving the theorem.

		Our approach is to devise an $\Adv$ so that the remainder of our reasoning resembles the (classical) proof of this theorem
		by Klivans and Sherstov~\cite{KlivansS09} (who in turn re-derived the proof of Kearns and Valiant~\cite{KearnsV94} in their paper). However the
		quantization of their proof is not straightforward.  
		Classically, it is possible to 	generate samples from the uniform distribution supported on $\mathcal{E}$ as follows: pick $r$
		and $b$ uniformly at random from their respective domains and then output
    the uniform classical examples $(e_{r,b},b)$. Quantumly, it is not clear how
    a polynomial-time adversary~$\Adv$ could create a quantum~example 	
		\begin{align}
		\label{eq:qexampleadv}
		\frac{1}{\sqrt{|\mathcal{E}|}} \sum_{r,b}\ket{e_{r,b}}\ket{b}.
		\end{align}
		Notice that if an adversary could prepare the state in
		\Cref{eq:qexampleadv}, then it could directly pass $t(n)$ quantum examples to
		the quantum-PAC learner (assumed in the theorem statement) and a similar
		analysis as in~\cite{KlivansS09} gives a quantum algorithm (or
    distinguisher) that satisfies \Cref{eq:quantumdistinguisher}. However, one issue while preparing the state in \Cref{eq:qexampleadv} is the following: the straightforward way to
		prepare the state  \Cref{eq:qexampleadv} would be create the uniform
		superposition $\frac{1}{\sqrt{|\mathcal{E}|}}\ket{r}\ket{b}\ket{0}$ and then in
		superposition prepare $\frac{1}{\sqrt{|\mathcal{E}|}}\ket{r}\ket{b}\ket{e_{r,b}}$.
		If there was a way to erase the first register, then we would get the state
		in \Cref{eq:qexampleadv}. However, erasing this register in a black-box way is equivalent
		to solving the index-erasure
    problem~\cite{ambainis:indexerasure,rosmansis:index}, which needs exponential~time.

		In order to circumvent the issue mentioned above, we now show how to tweak
		the proof and construct an adversary $\Adv_{\real}$ that correctly decrypt a
    challenge encryption (i.e., $e_{r^*,b^*}$) with
		non-negligible probability. Before constructing $\Adv_{\real}$, we first
		present an adversary $\Adv_{\ideal}$ that is able to create specific samples
		(see \Cref{eq:quantumsampleforlearn1})
		that $\Adv_{\real}$ could not create and analyze its success probability. 
		Then we show that the success probability of $\Adv_{\real}$ and $\Adv_{\ideal}$
    in predicting the challenge encryption (i.e., $e_{r^*,b^*}$) are very close.

		Let $S_{\ideal}=\{(r_i,b_i)_i\}$ be a fixed set of size $L=2n^ct(n)$ that
    contains $(r^*,b^*)$ and let us assume $\Adv_{\ideal}$ has access to $S_{\ideal}$.
		Then $\Adv_{\ideal}$ can create the quantum example
		\begin{align}
		\label{eq:quantumsampleforlearn1}
		\ket{\psi_{\ideal}}=\frac{1}{\sqrt{L}}\sum_{(r, b) \in
			S_{\ideal}}\ket{e_{r,b},b}
		\end{align}
		in polynomial time. 
		Notice that since $\Adv_{\real}$ does not know the decryption of $e_{r^*,b^*}$, so it
		would most likely not be able to obtain such an $S_{\ideal}$ and consequently would not
		be able to create 
		$\ket{\psi_{\ideal}}$ efficiently. We now continue our analysis for 	$\Adv_{\ideal}$  assuming it has access to quantum examples $\ket{\psi_{\ideal}}$. 
		$\Adv_{\ideal}$ first passes $t(n)$ copies of $\ket{\psi_{\ideal}}$ to the quantum-PAC learning
		algorithm for $\Cc$ (which is assumed by the theorem statement). 
		Let $\Hi(\De)$ be the distribution
		of the hypothesis output by the quantum-PAC learning algorithm when the samples
		come from the distribution $\De$. For a set $S = \{(r_i,b_i)\}_i$, let
		$\De_S$ be the uniform distribution on the training set $\{(e_{r,b},
		b):(r,b) \in S\}$.
		Suppose the PAC learning algorithm outputs a hypothesis $h\in \supp(\Hi(\De_S))$, then $\A_{\ideal}$ outputs $h(e_{r^*,b^*})$ as its guess for $b^*$. We now analyze the probability that output of $\Adv_{\ideal}$ is in fact the
		correct decryption of~$e^*$:
		\begin{align}
				\begin{aligned}
				\label{eq:hypothesisunderD2}
					\Pr_{e^*\in U_{E}}[\Adv_{\ideal}(e^*) = b^*] 
		&=\Pr_{e^* \in U_{E}}\Exp_{S \subseteq	E}\left[ \Adv_{\ideal}(e^*) = b^*
		\Big| e^* \in S \wedge |S| = L\right] \\
		&=\Pr_{e^* \in U_{E}}\Exp_{S \subseteq	E}\left[ \Pr_{h \sim \Hi(\De_{S})}[h(e^*) = b^*]
		\Big| e^* \in S \wedge |S| = L\right] \\
		&=\Pr_{e^*\in U_{E}}\sum_{\substack{S \subseteq
				E}}
		\Pr_{\textbf{S} \subseteq E}[\textbf{S} = S | e^* \in S \wedge |S| = L]
		\Pr_{h \sim \Hi(\De_{S})}[h(e^*) = b^*] \\
		&=\sum_{S \subseteq E}\sum_{e^* \in S}
		\Pr_{\textbf{S} \subseteq E}[\textbf{S} = S |
		|S| =L] \Pr_{W \in U_{S}}[W = e^*]\Pr_{h \sim \Hi(\De_{S})}[h(e^*) = b^*] \\
		&=\sum_{S \subseteq E}
		\Pr_{\textbf{S} \subseteq E}[\textbf{S} = S |
		|S| =L]\Pr_{h \sim \Hi(\De_{S})} \Big[\Exp_{e^* \in S} [h(e^*) = b^*]\Big] \\
		&\geq 
		\sum_{\substack{S \subseteq E}}
		\Pr_{\textbf{S} \subseteq E}[\textbf{S} = S |  |S| = L]
		\left(\frac{1}{2} + \frac{1}{n^c}\right)= \frac{1}{2} + \frac{1}{n^c}.
		\end{aligned}
		\end{align}
		In the first equality, the expectation is taken over uniformly random $S\subseteq E$ conditioned on $|S|=L$ and $e^*\in S$, second equality describes the operations of $\Adv_{\mathsf{ideal}}$,
		the fourth equality holds because 
		$$
		\Pr_{e^*\in U_{E}}\sum_{\substack{S \subseteq E}} \Pr_{\textbf{S} \subseteq E}[\textbf{S} = S | e^* \in S \wedge |S| = L]=\frac{1}{|E|}\binom{E-1}{L-1}^{-1} =
		\binom{E}{L}^{-1}\frac{1}{L}=\sum_{S \subseteq E}\sum_{e^* \in S}
		\Pr_{\textbf{S} \subseteq E}[\textbf{S} = S | |S| = L] \Pr_{W \in U_{S}}[W = e^*]
		$$ 
		and the inequality comes from the guarantees of the quantum learning
    algorithm.\footnote{Note that we have assumed for simplicity that (i)    $d_{\Kpriv}(e_{\Kpub,r}(b)) = b$ with probability $1$ and (ii) the quantum
    learning algorithm outputs a hypothesis with probability $1$. If we take into
    account the decryption probability error $\eps(n)$ and that the learning
    algorithm succeeded with probability $1-\delta$, then the lower bound in
    \Cref{eq:hypothesisunderD2} would be $\frac{1}{2} + (1-\delta)n^{-c}-
    \eps(n)$ and the remaining part of the analysis remains the same.} Hence
    \Cref{eq:hypothesisunderD2} shows that the ideal quantum algorithm
    $\Adv_{\ideal}$ satisfies the following: on input a uniformly random
    $e^*=e_{r^*,b^*}$ from the set $E$, $\Adv_{\ideal}$ outputs $b^*$ with probability at least $1/2+~n^{-c}$.

		We now consider $\Adv_{\real}$. Let
		\begin{align}
		\label{eq:quantumsampleforlearn2}
		\ket{\psi_{\real}}=\frac{1}{\sqrt{L-1}}\sum_{(r, b) \in
			S_{\real}}\ket{e_{r,b},b},
		\end{align}
		where $S_{\real} = S_{\ideal} \setminus \{e^*\}$. The idea is that for $L=2n^ct(n)$, the distance between the two states $\ket{\psi_{\ideal}}$ and $\ket{\psi_{\real}}$ is small. Also notice that
		the state in \Cref{eq:quantumsampleforlearn2} can be constructed efficiently by $\Adv_{\real}$:
		pick $L-1$ many $(r_i, b_i)$-pairs uniformly at random from
		their respective domains and set $S_{\real}=\{(r_i,b_i)_i\}$. From these classical values, the learner can create $t(n)$ many copies
		of the state $\ket{\psi_{\real}}$. We now show that a hypothesis sampled from $\Hi(\De_{S_{\real}})$ behaves ``almost similar" to a hypothesis sampled from $\Hi(\De_{S_{\ideal}})$.	For every $y\in \01^n$ and $b\in \01$, observe that
		\begin{align}
		\label{eq:realproblowerbound}
		\begin{aligned}
		\Pr_{h \sim \Hi(\De_{S_{\real}})}[h(y)=b] &= \Pr_{h \sim \Hi(\De_{S_{\real}})}[h(y)=b] -
		\Pr_{h \sim \Hi(\De_{S_{\ideal}})}[h(y)=b] +
		\Pr_{h \sim \Hi(\De_{S_{\ideal}})}[h(y)=b] \\
		&		\geq
		\Pr_{h \sim \Hi(\De_{S_{\ideal}})}[h(y)=b] -
		O\left(\frac{t(n)}{L}\right).
		\end{aligned}
		\end{align}

    The inequality holds since the statistical distance of the output of an
    algorithm $A$ in inputs $\ket{\phi}$ and $\ket{\phi'}$ is upper-bounded by
    $\norm{\ket{\phi} - \ket{\phi'}}^2$, and we have that
		\begin{align}
		\norm{\ket{\psi_{\ideal}}^{\otimes
				t(n)} - \ket{\psi_{\real}}^{\otimes t(n)}}^2 = t(n)\norm{\ket{\psi_{\ideal}} -
			\ket{\psi_{\real}}}^2 = O\left(\frac{t(n)}{L}\right)
		\end{align}
		by the definition of $\ket{\psi_{\ideal}}$ and $\ket{\psi_{\real}}$ in \Cref{eq:quantumsampleforlearn1} and~\ref{eq:quantumsampleforlearn2} respectively. We now analyze the probability that $\Adv_{\real}$ outputs the right decryption $b^*$, when given a uniformly random sample $e^*\in E$:	
		\begin{align}
		\label{eq:lowerboundsuccprobofreal}
		\begin{aligned}
	\Pr_{e^*\in U_{E}}[\Adv_{\real}(e^*) = b^*]	
		&=\Pr_{e^* \in U_{E}}\Exp_{S_{\real} \subseteq	E}\left[ \Adv_{\real}(e^*) = b^*
		\Big|  |S_{\real}| = L-1\right] \\
		&=\Pr_{e^* \in U_{E}}\Exp_{S_{\real} \subseteq	E}\left[ \Pr_{h \sim \Hi(\De_{S_{\real}})}[h(e^*) = b^*]
		\Big| |S_{\real}| = L=1\right] \\
		&\geq \Pr_{e^* \in U_{E}}\Exp_{S_{\real} \subseteq	E}\left[ \Pr_{h \sim
			\Hi(\De_{S_{\ideal}})}[h(e^*) = b^*]
		\Big| |S_{\real}| = L-1\right]- O\left(\frac{t(n)}{L}\right)\\
		&\geq \frac{1}{2}+\frac{1}{n^c}- O\left(\frac{t(n)}{L}\right)\geq \frac{1}{2}- O\left(\frac{1}{n^c}\right).
		\end{aligned}
		\end{align}
		
		In the equations above, recall that $S_{\ideal} = S_{\real} \cup \{(e^*,b^*)\}$, the first inequality used \Cref{eq:realproblowerbound}, the second inequality used the same analysis as in \Cref{eq:hypothesisunderD2} and the last inequality used $L=2n^ct(n)$. This
		shows that $\Adv_{\real}$, on input a uniformly random $e^*=\mathsf{Enc}_{\Kpub,
			r^*}(b^*)$ from the set $\mathcal{E}$, outputs~$b^*$ with probability at least
		$1/2+O(n^{-c})$.

		It remains to show that $\Adv$ satisfies \Cref{eq:quantumdistinguisher}, which is equivalent to
		\begin{align}\label{eq:quantumdistinguisher2}
		\vert 2\Pr_{K,r^*,
			b^*}[\Adv(\Kpub,e_{\Kpub,r^*}(b^*))=b^*]-1\vert.
		\end{align}
		Letting $\Adv$ be $\Adv_{\real}$ in the calculation above, we have 
		$$
		\vert 2\Pr_{K,r^*,
			b^*}[\Adv(\Kpub,e_{\Kpub,r^*}(b^*))=b^*]-1\vert\geq O\left(\frac{1}{n^c}\right)=\frac{1}{\poly(n)},
		$$
		where the inequality used \Cref{eq:lowerboundsuccprobofreal}. This concludes the proof of the theorem.
	\end{proof}
	\vspace{-1.5em}

\subsection{Hardness of learning \torp{$\TCz_2$}{TC0[2]} }
We show here a conditional hardness of quantum learning depth-$2$ threshold circuits.

\begin{restatable}{theorem}{learningtcztwo}
	\label{thm:hardnessoftc0two}
	Let $n \in \natural$. Assuming that the $n$-dimensional
	$\LWE$ problem is hard for a $\poly(n)$-time quantum computer, there is
	no $\poly(n)$-time weak quantum-PAC learner for the concept
	class of $\TCz_2$ circuits on $\widetilde{O}(n)$-bit inputs.
\end{restatable}
\vspace{-0.8em}

The main point of difference between the proof of
\Cref{thm:hardnessoftc0two} and the (classical) result of Klivans and
Sherstov~\cite{KlivansS09} is in the connection between \emph{quantum}
learning and public-key \emph{quantum} cryptosystems which we already
discussed in the previous section. The remaining part of our proof follows
their structure very closely. For brevity, we state a simple lemma from their paper without a proof. 

\begin{lemma}\cite[Lemma~4.1]{KlivansS09}
	\label{lem:halfspacesimpliesPTF}
	Fix $\varepsilon>0$ to be a constant.	Assume that the intersection of $n^\varepsilon$ light half-spaces is weakly quantum-PAC-learnable. Then for every constant $c>0$, the intersection of $n^c$ light degree-$2$ PTFs are weakly quantum-PAC learnable.\footnote{In~\cite[Lemma~4.1]{KlivansS09}, they state the classical version of this lemma. The exact same proof goes through when we replace classical learners by quantum learners.}
\end{lemma}

We are now ready to prove our main theorem.

\begin{proofof}{\Cref{thm:hardnessoftc0two}}
	In order to prove the hardness of $\TCz_2$, we first consider a subclass of $\TCz_2$, intersection of polynomially-many half-spaces, and prove the conditional quantum hardness of learning this subclass.

	Fix $\varepsilon>0$ to be a constant. Let $\Cc$ be the concept class of
	$n^\varepsilon$ light half-spaces and $\Cc'$ be the concept class of
	degree-$2$ light PTFs. Let us assume that $\Cc$ is quantum-PAC learnable.
	Then by \Cref{lem:halfspacesimpliesPTF}, the assumed learnability of
	$\Cc$ implies the quantum-PAC learnability of $\Cc'$. By
	\Cref{lem:SVPindeg2ptf}, the decryption function of the $\LWE$-cryptosystem
	is in $\Cc'$. Using \Cref{thm:cryptotolearn}, we now relate the
	quantum-PAC learnability of $\Cc'$ to the $\LWE$-cryptosystem as follows:
	suppose that there exists a quantum-PAC learning algorithm for $\Cc'$, then \Cref{thm:cryptotolearn} implies the existence of a distinguisher that can differentiate the encryptions of $0$ and $1$ in the $\LWE$-cryptosystem. As a consequence, by
	\Cref{thm:regevsvphardnesscrypto} this would result in a quantum
	polynomial-time distinguisher for~\LWE{}.
	Using Claim~\ref{claim:halfspaceintc0}, we have that $\Cc\subseteq \TCz_2$. Putting this together with our observation in the previous paragraph about quantum PAC learnability of $\Cc$, we get the theorem statement.
\end{proofof}%

\newcommand{\etalchar}[1]{$^{#1}$}

\appendix

\section{Time-efficient learning of \torp{$\NCz$}{NC0}}
\label{sec:nczeasy}
We prove the following theorem, which is straightforward from known results in quantum learning theory. Since we didn't find an explicit reference to the proof of such an explicit theorem in literature, we make it formal here.

\begin{theorem}
	\label{thm:learningdepthdcircuits}
	Let $c>0$ be a constant. Let $\Cc$ be the concept class of all Boolean functions on $n$ bits which are computable by polynomial number of gates with at most $2$ inputs and depth at most $d$. Then $\Cc$ can be learned under the uniform distribution with error at most $\varepsilon$, using at most $\widetilde{O}(2^d/\varepsilon)$ quantum examples, $O(2^{2^d})$ classical examples	and time $\widetilde{O}(n2^d/\varepsilon+2^{2^d}\log (1/\varepsilon))$.
\end{theorem}

The proof directly follows from {At\i c\i} and Servedio's theorem on learning $k$-juntas. %

	\begin{theorem}[{\cite{atici&servedio:testing}}]
		\label{thm:aticiservedio}
		There exists a quantum algorithm for learning $k$-juntas under the uniform distribution that uses ${O}((k\log k)/\varepsilon)$ uniform quantum examples, $O(2^{k})$ uniform classical examples, and $\widetilde{O}(nk/\varepsilon+2^{k}\log (1/\varepsilon))$ time.
	\end{theorem}
\vspace{-0.8em}

Now we proceed with proving our theorem. 

\begin{proofof}{\Cref{thm:learningdepthdcircuits}}%
Suppose that $f:\01^n\rightarrow \01$ is computed by a circuit with depth at most $d$. Since the gates of the circuit have fan-in at most $2$, clearly the output of the circuit for $f$ (on an arbitrary input $x$) depends only on $2^d$ bits of $x$. Hence every $f\in \Cc$ is a $2^d$-junta. Using \Cref{thm:aticiservedio}, it follows that the concept class of $2^d$-juntas can be learned using $\widetilde{O}(2^d/\varepsilon)$ quantum examples, $O(2^{2^d})$ classical examples and time $\widetilde{O}(n2^d/\varepsilon+2^{2^d}\log (1/\varepsilon))$.
\end{proofof}

\begin{corollary}
	Let $c>0$ be a constant.	Let $\Cc$ be the concept class of $n$-bit functions which are computable by polynomial number of gates with fan-in at most $2$, depth at most $\log (c\log n)$. Then $\Cc$ can be learned using at most $\widetilde{O}(n/\varepsilon)$ quantum examples, $O(n^c)$ classical examples and time $\widetilde{O}(n^c/\varepsilon)$.
\end{corollary}

\begin{proof}
	The proof immediately follows from \Cref{thm:learningdepthdcircuits} by plugging in $d=\log(c\log n)$.
\end{proof}
\vspace{-0.8em}

Observe that, since $\NCz$ is the class of circuits on $n$-bits with depth $O(1)$, $\NCz$ is contained in the concept class $\Cc$ considered in the corollary above.

\end{document}